\documentclass{article}
\usepackage[left=1in,right=1in,top=1in,bottom=1in]{geometry}
\usepackage{amsfonts,amsmath,mathrsfs,bbm,amsthm, amssymb, enumitem,comment,apptools}
\usepackage{xcolor}
\usepackage{authblk}
\usepackage{appendix}
\usepackage{enumitem}

\DeclareMathOperator{\id}{id}

\newcommand{\A}{\mathcal{A}}

\newcommand{\M}{\mathbb{M}}

\DeclareMathOperator{\diam}{diam}

\newcommand{\one}{\mathbbm{1}}

\newcommand{\E}{\mathscr{E}}

\DeclareMathOperator{\Tr}{Tr}

\newcommand{\lind}{\mathscr{L}}

 \def\idty{{\mathchoice {\mathrm{1\mskip-4mu l}} {\mathrm{1\mskip-4mu l}} %
{\mathrm{1\mskip-4.5mu l}} {\mathrm{1\mskip-5mu l}}}}

\let\emptyset\varnothing
\numberwithin{equation}{section}
\newtheorem{thm}{Theorem}[section]

\newtheorem{lem}[thm]{Lemma}
\newtheorem{cor}[thm]{Corollary}
\theoremstyle{definition}
\newtheorem{define}[thm]{Definition}

\newtheorem{rmk}[thm]{Remark}

\begin{document}
\title{On Quasi-Locality and Decay of Correlations for Long-Range Models of Open Quantum Spin Systems}
\author{Eric B. Roon}
\author{Robert 
Sims}

\affil{Department of Mathematics, The University of Arizona}
\date{\today}              
\renewcommand\Affilfont{\itshape\small}
\maketitle

\begin{abstract}
We consider models of open quantum spin systems with irreversible dynamics and show
that general quasi-locality results for long-range models, \textit{e.g.} as proven for the Heisenberg dynamics
associated to quantum systems in \cite{Matsuta}, naturally extend to this setting. 
Given these bounds, we provide two applications. First, we use these results to obtain 
estimates on a strictly local approximation of these finite-volume, irreversible dynamics.
Next, we show how these bounds can be used to estimate correlation decay in various states.

\end{abstract}
\tableofcontents

\section{Introduction}

Since the pioneering work of Lieb and Robinson, see \cite{LiebRobinson},  quasi-locality bounds for the dynamics of
models of quantum spin systems have become an indispensable tool in this many-body framework.
By now, the well-known list of applications is long, and we refer the interested reader to the books \cite{BratteliRobinson_2, Naaijkens} as well as the 
review-oriented article \cite{NachtergaeleSimsYoung} for a more complete set of references.

In the present work, we focus on studying the irreversible dynamics associated to open quantum systems which allow for dissipative interactions.
To our knowledge, the first quasi-locality results for such models were obtained in \cite{Hastings} and further extensions, as well as some applications, can be found in \cite{BarthelKliesch, Cubitt, KastoryanoEisert, Kliesch_CTT, Kliesch_LRB, Lucia, NachtergaeleVershyninaZagrebnov, Sweke}. 
Our main goal here is to extend the results of \cite{Matsuta}, where the authors consider the Heisenberg dynamics of quantum spin models with 
long-range interactions, to this setting of long-range models with an irreversible dynamics, defined as a quantum dynamical semigroup. 
In fact, our efforts were inspired by statements in the work of \cite{Sweke} where the authors indicate that they are unaware of such an extension in the literature.

Let us now briefly distinguish our results from that of some recent works. In the past few years, 
much attention has been given to understanding quasi-locality results for the dynamics of models
with long-range interactions. Interest in these bounds started with investigations of the Heisenberg dynamics 
associated to quantum spin Hamiltonians, see \cite{Else, Foss-Feig, KuwaharaSaito, Matsuta, Tran, Tran_2019}, and results for models with an irreversible dynamics
soon followed, see \cite{Guo, Sweke}.  

In contrast to the results of \cite{Sweke}, we derive a power-law Lieb-Robinson bound where, in a certain regime, the upper bound is linear in time.  In \cite[Theorem 3]{KastoryanoEisert}, the authors develop methods to obtain estimates on the spatial correlations of a dynamical fixed point. Our methods diverge from \cite{KastoryanoEisert}, in that we do not require a reversibility assumption on the local Lindbladians, and instead we estimate the correlations using ideas which go back to \cite{Poulin}. The authors of \cite{Guo} also require reversibility in their clustering estimates. Moreover, \cite{Guo} contains an iterative method to establish a power-law Lieb-Robinson bound whereas we use perturbative methods inspired by \cite{Matsuta}. Lastly, we emphasize that our results focus on models of open quantum systems defined over a discrete metric space. In contrast, we note that two recent papers \cite{Breteaux, Breteaux_preprint} generalize Lieb-Robinson bounds for open quantum systems to models in the continuum. 

The remainder of the paper is organized as follows. In Section~\ref{sec:set-up}, we introduce the 
basic set-up for the models we consider. Next, in Section~\ref{sec:fra}, 
we discuss a fundamental quasi-locality result for a class of models with finite-range interactions. We state this key estimate as Theorem~\ref{thm:appLRB}.  
Given this bound, we adapt the methods of \cite{Matsuta} to models with an irreversible dynamics.
We present this as two separate results. First, in Theorem~\ref{thm:dyn_diff}, 
we demonstrate an estimate for a finite-range approximation of the finite-volume dynamics.
Next, Theorem~\ref{thm:gen_LRB} establishes an estimate analogous to the result in \cite[Theorem 2.1]{Matsuta}.
In particular, a quasi-locality bound for the finite-volume dynamics is expressed in terms of
a quasi-locality bound for a finite-range approximation.  
Interestingly, our result in Theorem~\ref{thm:gen_LRB} differs from that of \cite[Theorem 2.1]{Matsuta}
in that we are able to recover the well-known (static) Lieb-Robinson bound of \cite[Theorem 2]{NachtergaeleVershyninaZagrebnov} in the limit as the finite-range parameter $R$ tends to infinity. Next, we supplement Theorem~\ref{thm:gen_LRB} with Theorem~\ref{thm:poly_dec_lrb} which 
provides an explicit estimate for models where only power-law decay of the interactions is known.
Then, in Section~\ref{sl_Approx}, we consider approximations of the finite-volume dynamics by a strictly local evolution. 
Such estimates are crucial in many applications, and we describe our results in
Theorem~\ref{thm:gen_sl_app} and Theorem~\ref{thm:sl_app_poly_dec}.
 
Finally, in Section~\ref{sec:prop_of_corr}, we estimate correlation decay in various states. First, in Section~\ref{sec:spatiallydecaying}, we follow
the techniques of \cite{NachtergaeleOgataSims} and demonstrate 
estimates on dynamic correlations for states with a quantified spatial decay, now in this case of an irreversible dynamics. 
Then, in Section~\ref{sec:dynamicalsteadystates}, we consider correlation decay for a class of rapidly mixing, dynamical fixed points. 
Unlike the general methods in \cite{KastoryanoEisert}, we here obtain our results 
with ideas that go back to earlier work of Poulin, see \cite{Poulin}. In an appendix, we 
relate our assumption of rapid mixing to other well-studied notions.   


\section{Locality}\label{Locality}
\subsection{Set-Up} \label{sec:set-up}

All the models we consider will be defined on a countable metric space.
In many applications, the models of interest arise naturally on $\mathbb{Z}^{\nu}$ where $\nu \in \mathbb{N}$, but
the precise structure of $\mathbb{Z}^{\nu}$ is unnecessary for most of the rigorous results discussed here. In general, we take 
$(\Gamma, d)$ to be a countable metric space. For any $X \subset \Gamma$, we will denote
the cardinality of $X$ by $|X|$ and the diameter of $X$ by ${\rm diam}(X) = \sup\{ d(x,y) : x,y \in X\}$.
Given $X \subset \Gamma$ and $r \geq 0$, it is convenient to a introduce an $r$-inflation of $X$ by setting 
$X(r) = \{ y \in \Gamma : d(y,x) \leq r \mbox{ for some } x \in X \}$. By $\mathcal{P}_0( \Gamma)$ , we will denote 
the set of finite subsets of $\Gamma$, \textit{i.e.} $\mathcal{P}_0(\Gamma) = \{ X \subset \Gamma : |X| < \infty \}$.

Let $(\Gamma, d)$ be a countable metric space. In order to introduce the models investigated here, we first associate to $(\Gamma, d)$ the 
structure of a quantum spin system. To do so, for each $x \in \Gamma$, we select a finite-dimensional (complex) Hilbert space of
states $\mathcal{H}_x$ and an algebra of observables $\mathcal{A}_x = \mathcal{B}( \mathcal{H}_x)$, the collection of 
bounded linear operators from $\mathcal{H}_x$ to itself. For any finite set $\Lambda \subset \Gamma$, the Hilbert space of
states associated with $\Lambda$ is then defined as the tensor product
\begin{equation} \label{local_hs}
\mathcal{H}_{\Lambda} = \bigotimes_{x \in \Lambda} \mathcal{H}_x 
\end{equation}
and similarly, the algebra of observables in $\Lambda$ is given by
\begin{equation} \label{local_alg}
\mathcal{A}_{\Lambda} = \bigotimes_{x \in \Lambda} \mathcal{A}_x \, .
\end{equation}
For finite sets $\Lambda_0 \subset \Lambda \subset \Gamma$, we identify $\mathcal{A}_{\Lambda_0}$
with the subalgebra $\mathcal{A}_{\Lambda_0} \otimes \idty_{\Lambda \setminus \Lambda_0}$ of $\mathcal{A}_{\Lambda}$ and
simply write $\mathcal{A}_{\Lambda_0} \subset \mathcal{A}_{\Lambda}$. An algebra of all strictly local observables $\mathcal{A}_{\Gamma}^{\rm loc}$ is then defined by the
inductive limit
\begin{equation}
\mathcal{A}_{\Gamma}^{\rm loc} = \bigcup_{\Lambda \in \mathcal{P}_0(\Gamma)} \mathcal{A}_{\Lambda}  \, .
\end{equation}
Given this, the $C^*$-algebra of quasi-local observables $\mathcal{A}_{\Gamma}$ is taken to be the
norm completion of $\mathcal{A}_{\Gamma}^{\rm loc}$.  The interested reader can find more details about this
construction \textit{e.g.} in \cite{BratteliRobinson_1, BratteliRobinson_2}. 

%
%
	For us, interactions will consist of a family $\mathscr{L} = \{L_Z: Z\subset \Gamma \text{ finite }\}$ of bounded linear maps $L_Z:\mathcal{A}_Z \to \mathcal{A}_Z$ in {\em Lindblad} form. Discovered independently in 1976 by \cite{Lindblad} and \cite{GoriniKossakowskiSudarshan}, the Lindblad equation characterizes generators of a norm-continuous, one parameter semigroup of unital completely positive maps. For us, this means there exists $H_Z = H_Z^*\in \mathcal{A}_Z$, a number $\ell_Z\in \mathbb{N}$ and operators $K_1, \dots, K_{\ell_Z}\in \mathcal{A}_Z$ so that for all $A\in \mathcal{A}_Z$, one has
	\begin{equation}\label{eqn:lindblad_equation}
		L_Z(A) = i[H_Z, A] + \sum_{j=1}^{\ell_Z} K_j^*AK_j - \frac{1}{2}\{K_j^*K_j, A\}.
	\end{equation} The total interaction in a given finite volume, $\Lambda$, is then the operator 
	\begin{equation}
		\mathscr{L}^{\Lambda} = \sum_{Z\subset \Lambda} L_Z.
	\end{equation}
	
The $\Lambda$-local dynamics generated by $\mathscr{L}^\Lambda$ form a norm-continuous, one-parameter semigroup of operators \begin{equation} T^{\Lambda}_t(A) = e^{t\mathscr{L}^{\Lambda}}(A), \quad A\in \mathcal{A}, \end{equation} which is well defined since $\mathscr{L}^{\Lambda}$ is bounded. The Lindblad form guaruntees that $T^{\Lambda}_t$ have the properties of a quantum dynamical semigroup. Namely, for each $t\in [0, \infty)$ and $\Lambda \in \mathscr{P}_0(\Gamma)$, the map $T_t^{\Lambda}$ is unit preserving and completely positive (ucp). We refer the interested reader to \cite{AlickiFannes, AlickiLendi} for more on the theory of quantum dynamical semigroups. 

Complete positivity is a general condition some maps between $C^*$-algebras possess. We recall that for a finite dimensional Hilbert space $H$, the one may regard $B(H) \otimes \M_n$ as $n\times n$, $B(H)$-valued block matrices.   Any linear map $\Phi:B(H) \to B(H)$ induces a linear map on $B(H)\otimes \M_n$ by acting block-wise \begin{equation} \Phi\otimes \one_{n\times n} ( [[X_{i,j}]]_{i,j=1}^n ) := [[ \Phi(X_{i,j})]]_{i,j=1}^n. \end{equation} The linear map $\Phi$ is said to be completely positive if for all $n$, the matrix $\Phi \otimes \one_{n\times n}(X^*X)$ is positive semidefinite for all $X\in B(H) \otimes \M_n$. 
Recall, by \cite[Proposition 3.6]{Paulsen}, that completely positive maps are also completely bounded in the sense that the quantity \begin{equation} \|\Phi\|_{cb}:= \sup_{n\in \mathbb{N}} \|\Phi\otimes \one_{n\times n}\| <\infty. \end{equation} In fact, every linear map between matrix algebras is completely bounded by \cite[Proposition 8.11]{Paulsen}. In particular, given a finite volume $X\subset \Lambda \in P_0(\Gamma)$, and a linear map $\Phi \in B(\mathcal{A}_X)$, one has
\begin{equation}
	\|\Phi \otimes \id_{\Lambda\setminus X}\|_{cb} = \|\Phi\|_{cb} \|\|\id_{\Lambda \setminus X}\|_{cb} = \|\Phi\|_{cb}\,,
\end{equation}  see, \textit{e.g.} \cite[Corollary 3.5.5]{BrownOzawa} for a detailed proof.

We describe a class of completely bounded maps for the purposes of our model.  Given a finite volume $Y$, the set $CB_0(Y)$ is defined to be the set of completely bounded maps $\Phi$ on $\mathcal{A}_Y$ satisfying $\Phi(\one_Y) = 0$. An example of such a map is given by fixing $B\in \mathcal{A}_Y$ and defining $\Phi(A) = [B, A]$ for $A\in \mathcal{A}_Y$.

%
%

We distinguish classes of models according to short or long range interactions by introducing decay functions, and note
that this notion is primarily of use in contexts where $\Gamma$ is infinite.
A non-increasing function $F:[0, \infty) \to (0, \infty)$ is said to be an $F$-function on $(\Gamma, d)$ if:

\noindent i) $F$ is uniformly summable, in the sense that
\begin{equation} \label{uni_sum}
\| F \| = \sup_{x \in \Gamma} \sum_{y \in \Gamma} F(d(x,y)) < \infty
\end{equation} 
and

\noindent ii) $F$ satisfies a convolution condition, \textit{i.e.} there is a positive number $C$ for which, given any pair $x,y \in \Gamma$,
\begin{equation} \label{conv_cond}
\sum_{z \in \Gamma} F(d(x,z)) F(d(z,y)) \leq C F(d(x,y)) 
\end{equation}
We define the convolution constant $C_F$ as the infimum over all constants $C$ for which (\ref{conv_cond}) holds.  
In general, one readily checks that if $F$ is an $F$-function on $(\Gamma, d)$, then for any $a \geq 0$, the function $F_a(r) = e^{-ar} F(r)$ is also
an $F$-function on $(\Gamma, d)$. Moreover, one has that $\| F_a \| \leq \| F \|$ and $C_{F_a} \leq C_F$. 

%
%

	The set of doubly anchored interactions are those for which there exists a constant $M$ so that for any $x,y\in \Gamma$,
	\begin{equation}
		\sum_{\substack{Z\in \mathcal{P}_0(\Gamma):\\ x,y\in Z}} \|L_Z\|_{cb} \le M F(d(x,y)). 
	\end{equation} In particular, the infimum of those $M$ defines a norm $\|\mathscr{L}\|_F$, which we call the $F$-norm.

As indicated above, many interesting quantum spin models are defined with $\Gamma = \mathbb{Z}^{\nu}$. In this case, the existence of an $F$-function
is clear. In fact, the function 
\begin{equation} \label{poly_F}
F(r) = \frac{1}{(1+r)^{\alpha}} \quad \mbox{for any } r \geq 0
\end{equation}
is an $F$-function on $(\mathbb{Z}^{\nu}, d)$ whenever $\alpha > \nu$. Clearly, the above function $F$ is summable (take for example $d(x,y) = |x-y|$ with the
$\ell^1$-distance), and a simple convexity argument shows that $C_F \leq 2^{\alpha} \| F \|$. 

Under some mild regularity assumptions, one can establish the existence of $F$-functions more generally. We say that a metric space
$(\Gamma, d)$ is $\nu$-regular if there are positive numbers $\kappa$ and $\nu$ for which
\begin{equation} \label{nu_reg}
\sup_{x \in \Gamma} |b_x(n+1)| \leq \kappa (n+1)^{\nu}  \quad \mbox{for any } n \geq  0.
\end{equation}
Here we have used $b_x(n)$ to denote the closed ball of radius $n$ centered at $x \in \Gamma$. 
Note that if $(\Gamma, d)$ is $\nu$-regular, then the function $F$ in (\ref{poly_F}) is an $F$-function on $(\Gamma, d)$ whenever $\alpha > \nu +1.$
In fact, for any $x \in \Gamma$ and any non-increasing, non-negative function $F$,
\begin{eqnarray}
\sum_{y \in \Gamma} F(d(x,y)) = F(0) + \sum_{n \geq 0} \sum_{\stackrel{y \in \Gamma:}{n< d(x,y) \leq n+1}} F(d(x,y)) & \leq & F(0) +  \sum_{n \geq 0} F(n) | b_x(n+1) \setminus b_x(n)| \nonumber  \\
& \leq &  F(0) +  \kappa \sum_{n \geq 0} F(n) (1+n)^{\nu} \label{F_est}
\end{eqnarray}
where the last bound uses that $(\Gamma, d)$ is $\nu$-regular. Taking $F$ as in (\ref{poly_F}), it is clear that $F$ is uniformly summable whenever
$\alpha > \nu +1$. As the previously indicated bound on the convolution constant does not depend on the metric space, the above $F$ is an $F$-function on
any $\nu$-regular $(\Gamma, d)$.

We note that, for many arguments, different regularity assumptions suffice. For example, the metric space $(\Gamma, d)$ is
said to be $\nu$-surface regular if there are positive numbers $\kappa$ and $\nu >1$ for which
\begin{equation} \label{nu_surf_reg} 
\sup_{x \in \Gamma} |b_x(n+1) \setminus b_x(n)| \leq \kappa (n+1)^{\nu-1}  \quad \mbox{for any } n \geq  0.
\end{equation}
Clearly, if $(\Gamma, d)$ is $\nu$-surface regular, then $(\Gamma, d)$ is $\nu$-regular (with the same value of $\nu$) as
\begin{eqnarray}
b_x(n+1) = \{ x \} \cup \bigcup_{m=0}^n \left( b_x(m+1) \setminus b_x(m) \right) \quad \mbox{implies} \quad |b_x(n+1)| & \leq & 1 + \kappa \sum_{m=0}^n (1+m)^{\nu -1} \nonumber \\
& \leq & (1 + \kappa)(1+n)^{\nu} \, .
\end{eqnarray}
There may be situations where these different regularity assumptions are useful (note that if $(\Gamma, d)$ is $\nu$-surface regular, then
the function $F$ in (\ref{poly_F}) is an $F$-function on $(\Gamma, d)$ whenever $\alpha > \nu$, as is the case when $\Gamma = \mathbb{Z}^{\nu}$)
however, we state our results in terms of the notion of $\nu$-regularity, as in (\ref{nu_reg}), as this is more general among metric spaces.


\subsection{Finite Range Approximations} \label{sec:fra}
Our first result, Theorem~\ref{thm:appLRB} below, presents the essential arguments of \cite{NachtergaeleVershyninaZagrebnov} from a 
slightly different perspective. Here we focus on an estimate corresponding to a finite-range approximation of the finite-volume dynamics. 
It is clear that our arguments allow us to recover a simplified version of Theorem 2 from \cite{NachtergaeleVershyninaZagrebnov} and we state this as
Theorem~\ref{thm:NVZ_LRB}. Moreover, if a given model is finite-range and uniformly bounded, then we obtain an estimate
with strong spatial decay; this is the content of Corollary~\ref{cor:frLRB} below. 
Next, we use Theorem~\ref{thm:appLRB} to derive an estimate on the difference between the finite-volume dynamics 
and this finite-range approximation; this is the content of Theorem~\ref{thm:dyn_diff}. Our argument here is similar to that of 
\cite{Matsuta} with two main differences. First, we work in the setting of open quantum systems.
Next, unlike the results in \cite{Matsuta}, our estimates, see specifically Theorem~\ref{thm:gen_LRB}, allow us to recover the bound from \cite{NachtergaeleVershyninaZagrebnov} by letting the finite range approximation parameter
tend to infinity. We end this section with Theorem~\ref{thm:poly_dec_lrb} which provides an explicit estimate for
a class of models with polynomially decay.

For the results proven below, we assume $(\Gamma, d)$ is a $\nu$-regular metric space equipped with an $F$-function $F$.
As described above, let $\mathcal{A}_{\Gamma}$ be a $C^*$-algebra of quasi-local observables defined over $(\Gamma, d)$ and take
$\mathscr{L} = \{ L_Z \}_{Z \in \mathcal{P}_0(\Gamma)}$ to be a dissipative interaction. For any $\Lambda \in \mathcal{P}_0(\Gamma)$, we denote 
the finite-volume generator and corresponding irreversible dynamics by
\begin{equation}
\mathscr{L}^{\Lambda} = \sum_{Z \subset \Lambda} L_Z \quad \mbox{and} \quad  T_t^{\Lambda} = {\rm exp}(t \mathscr{L}^{\Lambda}) \quad \mbox{for any } t \geq 0. 
\end{equation}
Analogously, we consider an approximate, finite-volume dynamics, for any $R>0$, by setting
\begin{equation}
\mathscr{L}^{\Lambda, R} = \sum_{\stackrel{Z \subset \Lambda:}{{\rm diam}(Z) \leq R}} L_Z \quad \mbox{and} \quad  T_t^{\Lambda, R} = {\rm exp}(t \mathscr{L}^{\Lambda, R}) \quad \mbox{for any } t \geq 0. 
\end{equation}

Throughout the rest of this section, we will deal extensively with $k$-tails of the exponential power series. As such we introduce the following notation for ease of reading: for all $t\ge 0$ and $k>0$, we write
\begin{equation}
	\E_t(k):= \sum_{n = \lceil k \rceil}^\infty \frac{t^n}{n!}.
\end{equation} We will often omit the ceiling function and instead write $n\ge k$ when no confusion can arise.

We now state and sketch a well-known proof of a locality estimate
in the special case of this finite-range approximation. 
%
%

\begin{thm}\label{thm:appLRB}
Let $(\Gamma, d)$ be a countable metric space equipped with an $F$-function $F$.
Let $\mathscr{L} = \{L_Z \}_{Z \in \mathcal{P}_0(\Gamma)}$ be a dissipative interaction with $\| \mathscr{L} \|_F< \infty$. 
Fix $X,Y \in \mathcal{P}_0(\Gamma)$ with $X \cap Y = \emptyset$. 
For any $\Lambda \in \mathcal{P}_0( \Gamma)$ with $X \cup Y \subset \Lambda$, $A \in \mathcal{A}_X$, and 
$K\in CB_0(Y)$ one has that 
    \begin{equation}\label{appLRB}
        \|K(T^{\Lambda, R}_t(A))\| \le \frac{\|K\|_{cb}\|A\|}{C_F} \E_{vt}\left(\frac{d(X,Y)}{R}\right) \sum_{x\in X}\sum_{y\in Y} F(d(x,y))
    \end{equation}
    for any $t \geq 0$ and $R>0$.  Here $v = \| \mathscr{L} \|_F C_F$. 
\end{thm}
\begin{proof}
 Consider the function $f:[0,\infty) \to \A_{\Lambda}$ given by
    \begin{equation} \label{def_f}
        f(t) = K(T^{\Lambda, R}_t(A)) \quad \text{for all $t\ge 0$.}
    \end{equation}
Clearly, one has that 
    \begin{align*}
        f'(t) = K(\mathscr{L}^{\Lambda, R}T^{\Lambda, R}_t(A)) &= K( \mathscr{L}^{\Lambda \setminus Y, R}\, T^{\Lambda, R}_t(A))\\
        &\quad + K( (\mathscr{L}^{\Lambda, R} - \mathscr{L}^{\Lambda \setminus Y, R})T^{\Lambda, R}_t(A))\\
        &= \mathscr{L}^{\Lambda \setminus Y, R}f(t) + K( (\mathscr{L}^{\Lambda, R} - \mathscr{L}^{\Lambda \setminus Y, R})T^{\Lambda, R}_t(A)).
    \end{align*} 
Regarding the above as a first-order, non-homogeneous linear equation, the unique solution of the corresponding initial value problem is 
   \begin{equation}
        f(t) = T^{\Lambda \setminus Y, R}_t(f(0)) + \int_0^t T^{\Lambda \setminus Y, R}_{t-s}(K ((\lind^{\Lambda, R}- \lind^{\Lambda \setminus Y, R})T^{\Lambda, R}_s(A))) \, ds.
    \end{equation}  
The norm bound 
    \begin{equation}\label{sol_bd}
        f(t) \le \|f(0)\| + \|K\|_{cb}\int_0^t  \|(\lind^{\Lambda, R} - \lind^{\Lambda \setminus Y, R})T^{\Lambda, R}_s(A)\| \, ds.
    \end{equation} 
readily follows. 

    For the convenience of iteration, introduce 
    \begin{equation}
        C_A(Z,t) = \sup_{\substack{G\in CB_0(Z):\\ G\neq 0}} \frac{\|GT^{\Lambda,R}_t(A)\|}{\|G\|_{cb}} \quad\text{ for any $Z\subset \Lambda$ and $t\ge 0$.}
    \end{equation}
   One checks that $C_A(Z,t)\le \|A\|$ for all $t\ge 0$, and moreover, 
    \[
        C_A(Z,0)\le \|A\| \delta_X(Z),
    \] where $\delta_X(Z) = 1$ if $Z\cap X \neq \emptyset$ and equals zero otherwise.  Notice that for each $s\ge 0$, we have 
    \[
        \|(\lind^{\Lambda,R}- \lind^{\Lambda \setminus Y, R})T^{\Lambda,R}_s(A)\| \le \sum_{\substack{Z\subset \Lambda: \\ Z\cap Y \neq \emptyset\\ \diam(Z)\le R}} \|L_ZT^{\Lambda, R}_s(A)\| \le \sum_{\substack{Z\subset \Lambda:\\ Z\cap Y \neq \emptyset,\\\diam(Z)\le R}}\|L\|_{cb}C_A(Z,s).
    \] It follows from~(\ref{sol_bd}) and the above that
    \[
        C_A(Y,t) \le C_A(Y,0) + \sum_{\substack{Z\subset \Lambda:\\ Z\cap Y\neq \emptyset,\\ \diam(Z)\le R}} \|L_Z\|_{cb}\int_0^t C_A(Z,s)ds.
    \]  Upon iteration, we find that 
    \begin{equation}
        C_A(Y,t) \le C_A(Y,0) + \|A\| \sum_{n=1}^\infty a_n \frac{t^n}{n!},
    \end{equation} where 
    \begin{equation}
        a_n = \sum_{\substack{Z_1\subset \Lambda:\\ Z_1\cap Y \neq \emptyset\\ \diam(Z_1)\le R}}\|L_{Z_1}\|_{cb} \sum_{\substack{Z_2\subset \Lambda:\\ Z_2\cap Z_1 \neq \emptyset\\ \diam(Z_2)\le R}}\|L_{Z_2}\|_{cb}\cdots \sum_{\substack{Z_n\subset \Lambda:\\ Z_n\cap Z_{n-1} \neq \emptyset\\ \diam(Z_n)\le R}}\|L_{Z_n}\|_{cb} \, \delta_X(Z_n).
    \end{equation} 
    The convergence of this iteration is easily justified, see \textit{e.g.} the proof of Theorem 3.1 in \cite{NachtergaeleSimsYoung}.  
    Next, one readily checks that $a_n = 0$ if $nR<d(X,Y)$.  Lastly, the general sums defining $a_n$ satisfy the estimate 
    \begin{equation}
        a_n \le \| \mathscr{L}\|_F^n C_F^{n-1} \sum_{x\in X}\sum_{y\in Y}F(d(x,y))\quad \text{for any $n\ge1$}
    \end{equation} 
 which can be argued as in the proof of Theorem 2.1 in \cite{NachtergaeleOgataSims}. 
    Since $C_A(Y,0)=0$, we have found that 
    \[
        C_A(Y,t)\le \frac{\|A\|}{C_F}\sum_{x\in X}\sum_{y\in Y}F(d(x,y)) \sum_{n\ge \frac{d(X,Y)}{R}}\frac{(\| \mathscr{L}\|_F C_F t)^n}{n!},
    \] from which~(\ref{appLRB}) follows. 
\end{proof}

%
%

The argument in the proof above follows closely the arguments for the proof of Theorem 2 in \cite{NachtergaeleVershyninaZagrebnov}. 
We have included this not only for completeness, but also because our later arguments use heavily 
these finite range approximations. In fact, we state the following simplified version of the result in \cite{NachtergaeleVershyninaZagrebnov}.

\begin{thm}[Theorem 2 in \cite{NachtergaeleVershyninaZagrebnov}] \label{thm:NVZ_LRB}
Let $(\Gamma, d)$ be a countable metric space equipped with an $F$-function $F$.
Let $\mathscr{L} = \{L_Z \}_{Z \in \mathcal{P}_0(\Gamma)}$ be a dissipative interaction with $\| \mathscr{L} \|_F< \infty$. 
Fix $X,Y \in \mathcal{P}_0(\Gamma)$ with $X \cap Y = \emptyset$. 
For any $\Lambda \in \mathcal{P}_0( \Gamma)$ with $X \cup Y \subset \Lambda$, $A \in \mathcal{A}_X$, and 
$K\in CB_0(Y)$ one has that 
    \begin{equation}\label{NVZ_LRB}
        \|K(T^{\Lambda}_t(A))\| \le \frac{\|K\|_{cb}\|A\|}{C_F} \left(e^{ \| \mathscr{L} \|_F C_F t} -1 \right) \sum_{x\in X}\sum_{y\in Y} F(d(x,y))
    \end{equation}
    for any $t \geq 0$.
\end{thm}

The proof of their results follows as in our proof of Theorem~\ref{thm:appLRB} excepting that one uses the full finite-volume
dynamics to define $f$ in (\ref{def_f}) and notes that, generally, $a_n \neq 0$ for any $n \geq 1$. 

In the event that a particular model is bounded and finite range, then Theorem~\ref{thm:appLRB} provides an estimate with strong spatial decay.
To make this precise, let us first introduce some terminology. Let $(\Gamma, d)$ be a countable metric space.
A dissipative interaction $\mathscr{L}= \{ L_Z \}_{Z \in \mathcal{P}_0(\Gamma)}$ on $(\Gamma, d)$ 
is said to be of finite range if there is some $R_0 \geq 0$ for which
\begin{equation} \label{finite_range}
L_Z = 0 \quad \mbox{for all } Z \in \mathcal{P}_0(\Gamma) \quad \mbox{with } {\rm diam}(Z) > R_0.
\end{equation} 
In this case, for any $\Lambda \in \mathcal{P}_0(\Gamma)$, the approximate dynamics eventually saturates in the sense that for any 
$R \geq R_0$, $T_t^{\Lambda,R}(A) = T_t^{\Lambda}(A)$ for all $A \in \mathcal{A}_{\Lambda}$ and $t \geq 0$.  
Furthermore, we say that a dissipative interaction $\mathscr{L}= \{ L_Z \}_{Z \in \mathcal{P}_0(\Gamma)}$ on 
$(\Gamma, d)$ is uniformly bounded if 
\begin{equation} \label{uni_bd}
\| \mathscr{L} \|_{\infty} = \sup_{Z \in \mathcal{P}_0(\Gamma)} \| L_Z \|_{cb} < \infty \, .
\end{equation}

%
%

\begin{cor} \label{cor:frLRB}
Let $(\Gamma, d)$ be a $\nu$-regular metric space and $\mathscr{L}= \{ L_Z \}_{Z \in \mathcal{P}_0(\Gamma)}$ a dissipative 
interaction on $(\Gamma, d)$ that is uniformly bounded and of finite range $R_0>0$. In this case,
for any $F$-function $F$ on $(\Gamma, d)$, 
\begin{equation} \label{frub_bd}
\| \mathscr{L} \|_F \leq \| \mathscr{L} \|_{\infty} \frac{2^{\kappa R_0^{\nu} -2}}{ F(R_0)} < \infty \, 
\end{equation}
and therefore, Theorem~\ref{thm:appLRB} applies for each $F$-function $F$ on $(\Gamma, d)$.  Moreover, for 
fixed $X,Y \in \mathcal{P}_0(\Gamma)$ with $X \cap Y = \emptyset$, $\Lambda \in \mathcal{P}_0( \Gamma)$ 
with $X \cup Y \subset \Lambda$, $A \in \mathcal{A}_X$, and $K\in CB_0(Y)$, one has that one has that
\begin{equation} \label{strong_LRB}
 \|K(T^{\Lambda}_t(A))\|  \leq   \frac{\|K\|_{cb}\|A\| |X| \| F \|}{C_F} (evt)^m e^{- m \ln(m) +vt}
\end{equation} 
for any $t \geq 0$. Here we have set $m = \lceil d(X,Y)/R_0\rceil$ and $v = \| \mathscr{L} \|_F C_F$ as before.
\end{cor}
\begin{proof}
Let $F$ be an $F$-function on $(\Gamma, d)$. In this case, for any $x,y \in \Gamma$ with $d(x,y) \leq R_0$,
\begin{equation} \label{fr_bd}
\sum_{\stackrel{Z \in \mathcal{P}_0(\Gamma):}{x,y \in Z}} \| L_Z \|_{cb}  = \sum_{\stackrel{Z \subset b_x(R_0):}{x,y \in Z}} \| L_Z \|_{cb} \leq \| \mathscr{L} \|_{\infty} 2^{|b_x(R_0)| -2}
\leq \| \mathscr{L} \|_{\infty} 2^{\kappa R_0^{\nu} -2}
\end{equation}
where the final bound uses $\nu$-regularity. Since the sum on the left-hand-side of (\ref{fr_bd}) vanishes whenever $d(x,y) > R_0$, the bound in (\ref{frub_bd}) is clear.

As a result, Theorem~\ref{thm:appLRB} applies for any choice of $F$ and yields
\begin{equation}
 \|K(T^{\Lambda}_t(A))\| =  \|K(T^{\Lambda, R_0}_t(A))\|  \leq  \frac{\|K\|_{cb}\|A\|}{C_F} \E_{vt}\left(\frac{d(X,Y)}{R_0}\right) \sum_{x\in X}\sum_{y\in Y} F(d(x,y)) \, .
\end{equation} 
By uniform summability, \textit{i.e.} (\ref{uni_sum}), it is clear that
\begin{equation}
\sum_{x \in X} \sum_{y \in Y} F(d(x,y)) \leq |X| \| F \| \, .
\end{equation}
The final bound claimed in (\ref{strong_LRB}) is the result of the following simple estimate. 
Let $k \in \mathbb{N}$. One readily checks that $k! > k^ke^{-k}$ holds. In this case, for any 
$t>0$,
\begin{equation}
\E_t(k) = \sum_{n=k}^{\infty} \frac{t^n}{n!} = e^t - \sum_{n=0}^{k-1}  \frac{t^n}{n!} \leq \frac{t^k}{k!} e^t  \leq t^k e^{-k \ln(k) + k +t} 
\end{equation}
and (\ref{strong_LRB}) follows. 
\end{proof}

%
%

We now use Theorem~\ref{thm:appLRB} to provide an estimate on the difference between the finite-volume dynamics and this finite-range approximation.  
For the estimate which follows, it is convenient to introduce a variant of our decay function. 
Recall that we are working on a countable metric space $(\Gamma, d)$ equipped with an $F$-function $F$. 
In terms of $F$, define a new function $G:[0, \infty) \to [0, \infty)$ by setting
\begin{equation} \label{def_G}
G(r) = \sup_{x \in \Gamma} \sum_{\stackrel{y \in \Gamma:}{d(x,y)>r}} F(d(x,y)) \, .
\end{equation}
Since $F$ is uniformly summable, as in (\ref{uni_sum}), it is reasonable to expect decay from $G$. 
When $(\Gamma, d)$ is $\nu$-regular, this can be seen in a fairly explicit manner. In fact, 
arguing similarly to (\ref{F_est}), one sees that for any $x \in \Gamma$ and $r>0$,
\begin{equation}
\sum_{\substack{y\in \Gamma:\\ d(x,y)> r}} F(d(x,y)) \leq \sum_{n= \lfloor r \rfloor}^\infty F(n) |b_x(n+1)\setminus b_x(n)|
\end{equation}
and so $\nu$-regularity implies that 
\begin{equation} \label{g_bd}
G(r) \leq \kappa \sum_{n = \lfloor r \rfloor}^{\infty} (1+n)^{\nu}F(n) \,.
\end{equation}
In this case, if $F$ has a finite $\nu$-th moment, \textit{i.e.} $\sum_{n \geq 0} (1+n)^{\nu}F(n) < \infty$, then one has that $\lim_{r \to \infty} G(r) =0$. 

%
%

\begin{thm}\label{thm:dyn_diff}
Let $(\Gamma, d)$ be a countable metric space equipped with an $F$-function $F$.
Let $\mathscr{L}= \{ L_Z \}_{Z \in \mathcal{P}_0(\Gamma)}$ be a dissipative interaction with $\| \mathscr{L} \|_F< \infty$. 
Fix $X,Y \in \mathcal{P}_0(\Gamma)$ with $X \cap Y = \emptyset$. 
For any $\Lambda \in \mathcal{P}_0( \Gamma)$ with $X \cup Y \subset \Lambda$ and $A \in \mathcal{A}_X$, one has that 
 \begin{equation}\label{dyn_diff_est}
        \|T^{\Lambda}_t(A) -T^{\Lambda, R}_t(A) \|  \leq \|A\| \| \mathscr{L} \|_F G(R/2) \left( t |X(r)| + \frac{\E_{vt}(1+ rR^{-1})}{v C_F} \,
        \sum_{x \in X} \sum_{y \in \Lambda \setminus X(r)} F(d(x,y)) \right)
    \end{equation} 
for any $t \geq 0$, $r \geq 0$, and $R>0$. Here $G$ is as in (\ref{def_G}) and $v = \| \mathscr{L} \|_F C_F$. 
\end{thm}
\begin{proof}
 For any $t \geq 0$, one has that
 \begin{equation} \label{dyn_diff}
        T^{\Lambda}_t(A) - T^{\Lambda, R}_t(A) = - \int_0^t \frac{d}{ds}T^{\Lambda}_{t-s}(T^{\Lambda,R}_s(A)) \, ds = \int_0^t T^{\Lambda}_{t-s}\left( (\lind^{\Lambda}-\lind^{\Lambda, R})(T^{\Lambda, R}_s(A))\right) \, ds \, .
  \end{equation}
  In this case, a simple norm bound implies 
    \begin{equation} \label{simple_norm_bd}
              \|T^{\Lambda}_t(A) -T^{\Lambda, R}_t(A) \|  \leq \int_0^t \|( \lind^{\Lambda}-\lind^{\Lambda, R})(T^{\Lambda, R}_s(A))\| \, ds \
    \end{equation}  
and for each $0\le s \le t$, the integrand above may be estimated as 
    \begin{equation} \label{integrand_bd}
        \|( \lind^{\Lambda}-\lind^{\Lambda, R})(T^{\Lambda, R}_s(A))\|\le \sum_{\substack{Z\subset \Lambda: \\ \diam(Z)> R}} \|L_Z T^{\Lambda, R}_s(A)\|.
    \end{equation}

We need only estimate the terms on the right-hand-side of (\ref{integrand_bd}). 
Before we begin, consider the following. Let $Z \subset \Lambda$ with ${\rm diam}(Z) >R$. We claim that if $x \in Z$, then
    there is some $y \in Z$ with $d(x,y) > R/2$. In fact, if this is not the case, then $Z \subset b_x(R/2)$ which contradicts ${\rm diam}(Z) >R$.
    
Now, introduce a parameter $r\ge 0$ describing a buffer region about $X$, the support of $A$. We find that
    \begin{align*}
        \sum_{\substack{Z\subset \Lambda:\\ \diam(Z)> R,\\ Z\cap X(r)\neq \emptyset}} \|L_ZT^{\Lambda, R}_s(A) \| \le \|A\| \sum_{x\in X(r)} \sum_{\substack{Z\subset \Lambda :\\ \diam(Z)> R,\\ x\in Z}} \|L_Z\|_{cb} &\le \|A\| \sum_{x\in X(r)} \sum_{\substack{y\in \Lambda:\\ 2d(x,y)>R}} \sum_{\substack{Z\subset \Lambda:\\ x,y\in Z}}\|L_Z\|_{cb}\\
        &\le \|A\| \|\mathscr{L}\|_F \sum_{x\in X(r)} \sum_{\substack{y\in \Lambda:\\ 2d(x,y)>R}} F(d(x,y)) \\
        & \le  \|A\| \|\mathscr{L}\|_F |X(r)| G(R/2).
    \end{align*} 
where we used the definition of $G$ in (\ref{def_G}).     

For the remaining terms on the right-hand-side of (\ref{integrand_bd}), an application of Theorem~\ref{thm:appLRB} shows that 
    \[
        \sum_{\substack{Z\subset \Lambda:\\ \diam(Z)> R\\ Z\cap X(r)=\emptyset}} \|L_ZT^{\Lambda, R}_s(A)\| \le \frac{\|A\|}{C_F} \sum_{\substack{Z\subset \Lambda:\\ \diam(Z)> R\\ Z\cap X(r) = \emptyset}} \|L_Z\|_{cb} \, \E_{vs}(d(X,Z)/R) \, \sum_{x\in X} \sum_{z \in Z}F(d(x,z)). 
    \]  
Since each set $Z$, as above, satisfies $Z \cap X(r) = \emptyset$, we have that 
\begin{equation}
\E_{vs}(d(X,Z)/R) \leq \E_{vs}(r/R) \quad \mbox{and moreover,} \quad \int_0^t \E_{vs}(r/R) \, ds = \frac{1}{v} \E_{vt} \left(\frac{r}{R} +1 \right) \, .
\end{equation}    
We need only estimate the remaining $s$-independent sum. To this end, fix $x \in X$ and note that
        \begin{align*}
            \sum_{\substack{Z\subset \Lambda:\\ \diam(Z)> R,\\ Z\cap X(r) = \emptyset}} \|L_Z\|_{cb} \sum_{z \in Z} F(d(x,z)) &\le \sum_{z\in \Lambda \setminus X(r)} F(d(x,z)) 
            \sum_{\substack{Z \subset \Lambda: \\ \diam(Z) >R, \\ z \in Z, Z \cap X(r) = \emptyset} }\| L_Z \|_{cb} \\
            & \le  \sum_{z\in \Lambda \setminus X(r)} F(d(x,z))    \sum_{\substack{y\in \Lambda \setminus X(r):\\ 2d(z,y)>R}}  \sum_{\substack{Z \subset \Lambda:\\ z,y \in Z}}\|L_Z\|_{cb}\\
            &\le \|\mathscr{L}\|_F \sum_{z\in \Lambda \setminus X(r)}  F(d(x,z))  \sum_{\substack{y\in \Lambda \setminus X(r):\\ 2d(z,y)>R}} F(d(z,y))\\
            &\le \| \mathscr{L}\|_F G(R/2) \sum_{z\in \Lambda \setminus X(r)} F(d(x,z)) \, .
        \end{align*} 
The estimate claimed in (\ref{dyn_diff_est}) now follows. 
\end{proof} 

%

An immediate consequence of Theorem~\ref{thm:dyn_diff} is the following. 

\begin{thm}\label{thm:gen_LRB}
Let $(\Gamma, d)$ be a countable metric space equipped with an $F$-function $F$.
Let $\mathscr{L}= \{ L_Z \}_{Z \in \mathcal{P}_0(\Gamma)}$ be a dissipative interaction with $\| \mathscr{L} \|_F< \infty$. 
Fix $X,Y \in \mathcal{P}_0(\Gamma)$ with $X \cap Y = \emptyset$. 
For any $\Lambda \in \mathcal{P}_0( \Gamma)$ with $X \cup Y \subset \Lambda$, $A \in \mathcal{A}_X$, and 
$K\in CB_0(Y)$ one has that 
 \begin{eqnarray}\label{LRB}
        \|K(T^{\Lambda}_t(A))\|  & \leq & \| K(T^{\Lambda, R}_t(A)) \| + \\
        & \mbox{ } & \quad +  \|K\|_{cb} \|A\| \| \mathscr{L} \|_F G(R/2) \left( t |X(r)| + \frac{\E_{vt}(1+ rR^{-1})}{v C_F} \,\sum_{x \in X} \sum_{y \in \Lambda \setminus X(r)} F(d(x,y))  \right) \nonumber
    \end{eqnarray} 
for any $t \geq 0$, $r \geq 0$, and $R>0$. 
\end{thm}
\begin{proof}
   Apply $K$ to both sides of (\ref{dyn_diff}). A norm bound shows that 
    \begin{equation}
        \|K(T^{\Lambda}_t(A))\| \le \|K(T^{\Lambda,R}_t(A))\| + \|K\|_{cb}\int_0^t \|( \lind^{\Lambda}-\lind^{\Lambda, R})(T^{\Lambda, R}_s(A))\| \, ds.
    \end{equation} 
As is clear from (\ref{simple_norm_bd}), the integral above is precisely the quantity estimated in the previous proof.
The bound in (\ref{LRB}) follows.
\end{proof} 

Although we have already remarked that the bound proven in Theorem~\ref{thm:NVZ_LRB} can be obtained
directly by reframing the arguments in the proof of Theorem~\ref{thm:appLRB}, this result also follows from the bounds established in Theorem~\ref{thm:dyn_diff} and/or Theorem~\ref{thm:gen_LRB} above. 
In fact, since $d(X,Y)/R>0$, one has that
\begin{equation}
\| K(T^{\Lambda, R}_t(A)) \| \leq \frac{\|K\|_{cb}\|A\|}{C_F} \left(e^{ v t} -1 \right) \sum_{x\in X}\sum_{y\in Y} F(d(x,y)
\end{equation}
uniformly in $R$. Moreover, if the $F$-function on $(\Gamma, d)$ has a finite $\nu$-th moment, then $\lim_{R \to \infty} G(R/2) =0$. 
Thus for $r\geq 0$ and $t \geq 0$ fixed, Theorem~\ref{thm:NVZ_LRB} follows from Theorem~\ref{thm:gen_LRB} by taking the $\limsup_{R \to \infty}$ on both sides 
of (\ref{LRB}). In fact, one can also see this using Theorem~\ref{thm:appLRB}, Theorem~\ref{thm:dyn_diff}, and a continuity-type argument.

We now investigate an application of Theorem~\ref{thm:gen_LRB} which does not follow from Theorem~\ref{thm:NVZ_LRB}.
In particular, we consider models for which we only assume power-law decay. 

%
%
\begin{thm}\label{thm:poly_dec_lrb}
Let $(\Gamma, d)$ be $\nu$-regular and $\mathscr{L} = \{L_Z \}_{Z \in \mathcal{P}_0(\Gamma)}$ be a dissipative interaction for which
$\| \mathscr{L} \|_F< \infty$ with the choice of $F$-function given by
\begin{equation} \label{power_law_F}
F(x) = \frac{1}{(1+x)^{\alpha}} \quad \mbox{for all } x \geq 0
\end{equation}
with some $\alpha > 2\nu +1$. Fix $\epsilon \in (0, \alpha - 2 \nu -1)$ and set $\alpha_{\epsilon} = \alpha - \nu - 1- \epsilon$.
Let $0 < \delta < 1$ satisfy $(1- \delta)\alpha_{\epsilon} > \nu$. 

Fix $X,Y \in \mathcal{P}_0(\Gamma)$ with $X \cap Y = \emptyset$, $\Lambda \in \mathcal{P}_0( \Gamma)$ with $X \cup Y \subset \Lambda$, 
$A \in \mathcal{A}_X$, and $K\in CB_0(Y)$. In this case, there $C>0$ for which 
\begin{equation} \label{power_law_lrb}
\| K (T^{\Lambda}_t(A)) \| \leq  C \| K \|_{cb} \| A \| |X| \cdot \frac{t}{(1+ d(X,Y))^{(1- \delta) \alpha_{\epsilon} - \nu}}
\end{equation}
whenever $d(X,Y) \geq 1$ and $0 \leq evt \leq d(X,Y)^{\delta}$. Moreover, one may take
\begin{equation} \label{constant} 
C =  \kappa C_{\epsilon} \| \mathscr{L} \|_F  \left( e + 2^{2\alpha_{\epsilon}} 2^{1- \delta} \left( \frac{C_{\epsilon} + C_F}{C_F} \right) \right) \, .
\end{equation}
\end{thm}

\begin{proof}
As we discussed previously, $\nu$-regularity guarantees that the function $F$ in (\ref{power_law_F}) is an $F$-function on 
$(\Gamma, d)$ whenever $\alpha > \nu +1$.  Since we have taken $\alpha > 2 \nu +1$, this is certainly the case. Moreover, by our choice of 
$\epsilon$, it is clear that $\alpha_{\epsilon} = \alpha - \nu - 1 - \epsilon > \nu$, and thus there are choices of $0< \delta <1$ for which
$(1 - \delta) \alpha_{\epsilon} > \nu$. With this understanding, observe that under these conditions
Theorem~\ref{thm:gen_LRB} applies, and the bound given in (\ref{LRB}) holds. We regard the right-hand-side of this estimate
as three terms and consider each separately. For notational convenience, in the calculations below we set $d=d(X,Y)$.  

For the first term, we apply Theorem~\ref{thm:appLRB} and find that 
\begin{eqnarray} \label{term1_bd}
\| K(T_t^{\Lambda, R}(A)) \| & \leq & \frac{ \| K \|_{cb} \| A \|}{C_F} \E_{vt} \left( \frac{d}{R} \right) \sum_{x \in X} \sum_{y \in Y} F(x(x,y)) \nonumber \\
& \leq &  \| K \|_{cb} \| A \|  \| \mathscr{L} \|_F e^{- d/R} t e^{evt+1} \frac{ \kappa C_{\epsilon} |X|}{(1+d)^{\alpha_\epsilon}}
\end{eqnarray}
Here we have used that for any $X,Y \in \mathcal{P}_0(\Gamma)$ with $X \cap Y = \emptyset$, $R>0$ and $t\geq0$, the bounds:
\begin{equation} \label{F_sum}
\sum_{x \in X} \sum_{y \in Y} F(d(x,y)) \leq \frac{ \kappa C_{\epsilon} |X|}{(1+d)^{\alpha_\epsilon}} \quad \mbox{with} \quad C_{\epsilon} = \sum_{n \geq 0} \frac{1}{(1+ n)^{ 1+ \epsilon}}< \infty \, 
\end{equation}
and 
\begin{equation} \label{E_bd}
\E_{vt} \left( \frac{d}{R} \right) = \sum_{n \geq d/R} \frac{(vt)^n}{n!}  \leq  \sum_{n \geq d/R} \frac{(vt)^n}{n!} e^{n - d/R} 
 \leq  e^{- d/R} \left(e^{evt}-1 \right) \leq v e^{-d/R} te^{evt+1}
\end{equation}
hold and recall that $v = \| \mathscr{L} \|_F C_F$. (\ref{term1_bd}) follows. 

For the second term, the non-constant part may be estimated as
\begin{equation} \label{term2_bd}
G(R/2) \cdot t \cdot |X(r)| \leq \frac{\kappa C_{\epsilon}}{(1+ \lfloor R/2 \rfloor)^{\alpha_\epsilon}} \cdot t \cdot |X| \kappa r^{\nu}
\end{equation}
Here we have used (\ref{g_bd}) to estimate $G$ with $C_{\epsilon}$ as in (\ref{F_sum}). In addition, we also used that for $X \in \mathcal{P}_0(\Gamma)$,
\begin{equation}
X(r) = \bigcup_{x \in X} b_x(r) \quad \mbox{implies} \quad |X(r)| \leq \sum_{x \in X} |b_x(r)| \leq |X| \kappa r^{\nu} 
\end{equation}
whenever $r\geq1$. The final bound above uses $\nu$-regularity. 

For the third term, the non-constant part may be estimated as
\begin{equation} \label{term3_bd}
G(R/2) \cdot \E_{vt}(1+r/R) \cdot \sum_{x \in X} \sum_{\Lambda \setminus X(r)} F(d(x,y)) \leq \frac{\kappa C_{\epsilon}}{(1+ \lfloor R/2 \rfloor)^{\alpha_\epsilon}} \cdot
e^{- r/R} v t e^{evt} \cdot \frac{ \kappa C_{\epsilon} |X|}{(1+r)^{\alpha_\epsilon}}
\end{equation}

Now, comparing the time-independent exponential factors in the estimates from (\ref{term1_bd}) and (\ref{term3_bd}), set the 
parameters $R = d^{1- \delta}$ and $r = d$.  
In this case,
\begin{equation}
e^{- d/R} = e^{-d^{\delta}} = e^{-r/R}\, .
\end{equation}
With $M = \| K \|_{cb}\| \mathscr{L} \|_F \kappa C_{\epsilon}  |X| \| A \|$ and this choice of parametrization, we have shown that
\begin{eqnarray} \label{big_est}
\| K(T_t^{\Lambda}(A)) \| & \leq & M \left( \frac{e^{- d^{\delta}} t e^{evt+1}}{(1+d)^{\alpha_{\epsilon}}} 
+  \frac{\kappa t d^{\nu}}{(1+ \lfloor d^{1- \delta}/2 \rfloor)^{\alpha_\epsilon}} + \frac{\kappa C_{\epsilon} e^{- d^{\delta}} t e^{evt}}{ C_F(1+ \lfloor d^{1- \delta}/2 \rfloor)^{\alpha_{\epsilon}} (1+d)^{\alpha_{\epsilon}}}  \right) \nonumber \\
& \leq & M t \left( \frac{e e^{- d^{\delta} +evt}}{(1+d)^{\alpha_{\epsilon}}} 
+ \kappa 2^{2 \alpha_{\epsilon}} \left( \frac{(1+d)^{\nu}}{(1+ d^{1- \delta})^{\alpha_\epsilon}} + \frac{C_{\epsilon} e^{- d^{\delta} +evt}}{ C_F (1+  d^{1- \delta})^{\alpha_{\epsilon}} (1+d)^{\alpha_{\epsilon}}} \right) \right) \nonumber \\
& \leq & M t \left( \frac{e e^{- d^{\delta} +evt }}{(1+d)^{\alpha_{\epsilon}}} 
+ \kappa 2^{2 \alpha_{\epsilon}} 2^{1- \delta} \left( \frac{1 }{(1+ d)^{(1- \delta)\alpha_\epsilon - \nu}} + \frac{C_{\epsilon} e^{- d^{\delta} +evt}}{ C_F(1+d)^{(2-\delta) \alpha_{\epsilon}}} \right) \right)
\end{eqnarray}
Note that to simplify the estimate above, we used two basic bounds. First, we used
\begin{equation}
\frac{1}{(1+ \lfloor R/2 \rfloor)^{\alpha_{\epsilon}}} \leq \frac{2^{2 \alpha_{\epsilon}}}{(1+R)^{\alpha_{\epsilon}}} \quad \mbox{valid when } R>0 \mbox{ and } \alpha_{\epsilon} >0. 
\end{equation}
Next, we used
\begin{equation}
\frac{1}{(1+d^{a})} \leq \frac{2^{a}}{(1+d)^{a}} \quad \mbox{whenever } d \geq 1 \mbox{ and } a >0 \, .
\end{equation}
As the final powers appearing on the right-hand-side of (\ref{big_est}) are ordered: $(2 - \delta) \alpha_{\epsilon} > \alpha_{\epsilon} > (1- \delta) \alpha_{\epsilon} - \nu$, the bound
claimed in (\ref{power_law_lrb}), as well as (\ref{constant}), readily follows. 
\end{proof}

%

\subsection{Strictly Local Approximations}\label{sl_Approx}

In this section, we investigate a different approximation for these finite-volume dynamics. 
Here, we provide explicit estimates for comparisons to a strictly local dynamics. 
Theorem~\ref{thm:gen_sl_app} below gives a general estimate for dissipative interactions which have a
finite $F$-norm. Such a bound has found many useful applications, particularly in cases where
the interaction is known to have a finite $F$-norm with a function $F$ which decays exponentially fast.
We supplement this result with Theorem~\ref{thm:sl_app_poly_dec} which focuses on the case where
the interactions are only known to decay like a power-law. The decay we obtain here is identical to that of 
Theorem~\ref{thm:poly_dec_lrb}. Before proving these results, we first establish a technical estimate in 
Lemma~\ref{lem:NOS}.

We begin with a technical estimate, see Lemma~\ref{lem:NOS} below, and note that similar bounds have been used \textit{e.g.} in \cite{NachtergaeleOgataSims}. 
To state this precisely, we first observe that for certain combinatorial sums, it is convenient to define 
a set of surface sets. Specifically, let $(\Gamma, d)$ be a countable metric space. 
For any $X \subset \Lambda \subset \Gamma$, define the set of $X$-surface sets in $\Lambda$ by setting
\begin{equation}
S_X^{\Lambda} = \{ Z \subset \Lambda \, | \, Z \cap X \neq \emptyset \mbox{ and } Z \cap ( \Lambda \setminus X) \neq \emptyset \} \, .
\end{equation}
The following bound holds.  
\begin{lem}\label{lem:NOS}
Let $(\Gamma, d)$ be a countable metric space equipped with an $F$-function $F$ and $\mathscr{L}= \{ L_Z \}_{Z\in \mathcal{P}_0(\Gamma)}$ a dissipative interaction with $\| \mathscr{L} \|_F< \infty$. Fix $X \in \mathcal{P}_0(\Gamma)$. For any $x \in X$ and $\Lambda \in \mathcal{P}_0(\Gamma)$ with $X \subset \Lambda$, one has that 
\begin{equation}\label{base_est_1}
\sum_{\substack{Z\in S_{X(r)}^\Lambda\\ Z\cap X = \emptyset}} \|L_Z\|_{cb} \sum_{z\in Z} F(d(x,z)) \le \| \mathscr{L}\|_{F}(C_{F} + \|F\|) \sum_{y\in \Lambda\setminus X(r)} F(d(x,y)).
\end{equation}
for any $r \geq 0$. 
\end{lem}
\begin{proof}
Note that for any set $Z \subset \Lambda$, we may write 
\begin{equation} \label{decomp}
\sum_{z \in Z} \cdot = \sum_{z \in Z\cap X(r)} \cdot  \quad+ \quad \sum_{z \in Z \cap (\Lambda \setminus X(r))} \cdot 
\end{equation} 
We use this decomposition to prove (\ref{base_est_1}). Observe that those sites which are $r$-close to $X$ satisfy
\begin{eqnarray} \label{z_close_1} 
\sum_{\stackrel{Z \in S^{\Lambda}_{X(r)}:}{Z \cap X  = \emptyset}} \| L_Z \|_{cb} \sum_{z \in Z \cap X(r)} F(d(x,z)) & \leq &\sum_{y \in \Lambda \setminus X(r)}  \sum_{z \in X(r) \setminus X} F(d(x,z)) \sum_{\stackrel{Z \subset \Lambda :}{z,y \in Z }} \| L_Z \|_{cb}  \nonumber \\
& \leq & \| \mathscr{L} \|_F \sum_{y \in \Lambda \setminus X(r)}  \sum_{z \in X(r) \setminus X} F(d(x,z)) F(d(z,y)) \nonumber \\
& \leq & \| \mathscr{L} \|_F C_F \sum_{y \in \Lambda \setminus X(r)} F(d(x,y)) 
 \end{eqnarray}
Similarly, for those sites $r$-far from $X$, we have that
 \begin{eqnarray} \label{z_far_1}
\sum_{\stackrel{Z \in S^{\Lambda}_{X(r)}:}{Z \cap X = \emptyset}} \| L_Z \|_{cb} \sum_{z \in Z \cap ( \Lambda \setminus X(r))} F(d(x,z)) & \leq & \sum_{y \in X(r) \setminus X} \sum_{z \in \Lambda \setminus X(r)}  F(d(x,z)) \sum_{\stackrel{Z \subset \Lambda :}{z,y \in Z }} \| L_Z \|_{cb}  \nonumber \\
& \leq & \| \mathscr{L} \|_F \sum_{z \in \Lambda \setminus X(r)}  F(d(x,z)) \sum_{y \in X(r) \setminus X}  F(d(z,y)) \nonumber \\
& \leq & \| \mathscr{L} \|_F \| F \| \sum_{z \in \Lambda \setminus X(r)} F(d(x,z)) 
 \end{eqnarray}
Using (\ref{decomp}), one sees that (\ref{base_est_1}) follows from (\ref{z_close_1}) and (\ref{z_far_1}). 
\end{proof}

%
    
    \begin{thm}\label{thm:gen_sl_app}
    Let $(\Gamma, d)$ be a countable metric space equipped with an $F$-function $F$ and $\mathscr{L}= \{ L_Z \}_{Z \in \mathcal{P}_0(\Gamma)}$ a dissipative interaction with $\| \mathscr{L} \|_F< \infty$. Fix $X \in \mathcal{P}_0(\Gamma)$ and $A \in \mathcal{A}_X$. For any $\Lambda \in \mathcal{P}_0(\Gamma)$ with $X \subset \Lambda$, 
   \begin{equation}\label{localize}
            \|T^\Lambda_t(A) - T^{X(r)}_t(A)\| \le \|A\| \| \mathscr{L}\|_{F} \left(t + \frac{C_{F} + \| F \|}{C_{F}} \int_0^t (e^{vs} -1) \, ds \right) \sum_{x \in X} \sum_{y \in \Lambda \setminus X(r)} F(d(x,y)) 
            \end{equation}   
    for any $t \geq 0$ and $r \geq 1$. Here $v = \| \mathscr{L} \|_F C_F$. 
    \end{thm}
    \begin{proof}
We argue as in the beginning of the proof of Theorem~\ref{thm:dyn_diff}. Note that for any $t \geq 0$, 
    \[
        T^{\Lambda}_t(A) - T^{X(r)}_t(A) = - \int_0^t \frac{d}{ds}T^{\Lambda}_{t-s}(T^{X(r)}_s(A)) \, ds = \int_0^t T^{\Lambda}_{t-s}\left( (\lind^{\Lambda}-\lind^{X(r)})(T^{X(r)}_s(A))\right) \, ds.
    \] 
An application of the triangle inequality shows that
\begin{equation}\label{eqn:gen_sl_app_1}
 \left\|T^{\Lambda}_t(A)-T^{X(r)}_t(A) \right\| \leq \int_0^t \left\|  \left( \mathscr{L}^{\Lambda} - \mathscr{L}^{X(r)} \right) T^{X(r)}_s(A) \right\| \, ds \le \sum_{Z \in S^{\Lambda}_{X(r)}} \int_0^t \|L_Z T^{X(r)}_s(A)\| \, ds
\end{equation}
For those terms $Z$ which intersect $X$, one has
 \begin{eqnarray}
  \sum_{\substack{Z\in S^\Lambda_{X(r)}:\\ Z\cap X \neq \emptyset}} \int_0^t \|L_Z T^{X(r)}_s(A)\| \, ds & \leq & t\|A\|   \sum_{\substack{Z\in S^\Lambda_{X(r)}:\\ Z\cap X \neq \emptyset}}  \| L_Z \|_{cb} \nonumber \\
  & \leq & t \| A \| \sum_{x\in X} \sum_{y\in \Lambda \setminus X(r)} \sum_{\stackrel{Z \subset \Lambda:}{x,y \in Z}}\|L_Z\|_{cb} \nonumber \\
  & \leq & t\|A\| \| \mathscr{L}\|_F \sum_{x \in X} \sum_{y\in \Lambda \setminus X(r)}F(d(x,y)) \, .
  \end{eqnarray}
For those terms $Z$ that do not intersect $X$, Theorem~\ref{thm:NVZ_LRB} applies and 
\begin{align*}
            \sum_{\substack{Z \in S^\Lambda_{X(r)}:\\ Z\cap X = \emptyset}} \int_0^t \|L_Z T^{X(r)}_s(A)\| \, ds &\le \frac{\|A\|}{C_{F}} \int_0^t (e^{vs} - 1) \, ds 
            \sum_{\substack{Z \in S_{X(r)}^{\Lambda}\\ Z\cap X = \emptyset}} \|L_Z\|_{cb}\sum_{x\in X}\sum_{z\in Z}F(d(x,z))\\
            &\le \frac{\|A\|}{C_{F}} \int_0^t (e^{vs} - 1) \, ds \,  \| \mathscr{L}\|_{F} (C_F + \|F\|) \sum_{x \in X} \sum_{y \in \Lambda \setminus X(r)} F(d(x,y)),
\end{align*} 
where, for the final bound above,  we used Lemma~\ref{lem:NOS}. This completes the proof. 
\end{proof}

%

The final result of this section, Theorem~\ref{thm:sl_app_poly_dec} below, 
also provides an estimate on an approximation of the finite-volume dynamics with a strictly local dynamics. In this case, however, we restrict our attention to models generated by interactions which are only known to decay by a power-law. Here, in order to obtain a bound analogous to 
the one found in Theorem~\ref{thm:poly_dec_lrb}, we find it useful to argue differently than we did in the proof of
Theorem~\ref{thm:gen_sl_app}.

%

 \begin{thm}\label{thm:sl_app_poly_dec}
Let $(\Gamma, d)$ be $\nu$-regular and $\mathscr{L} = \{L_Z \}_{Z \in \mathcal{P}_0(\Gamma)}$ be a dissipative interaction with
$\| \mathscr{L} \|_F< \infty$ for the choice of $F$-function
\begin{equation} \label{power_law_F_2}
F(x) = \frac{1}{(1+x)^{\alpha}} \quad \mbox{for all } x \geq 0
\end{equation}
with some $\alpha > 2\nu +1$. Fix $\epsilon \in (0, \alpha - 2 \nu -1)$, set $\alpha_{\epsilon} = \alpha - \nu - 1- \epsilon$, and 
let $0 < \delta < 1$ satisfy $(1- \delta)\alpha_{\epsilon} > \nu$. 

Fix $X \in \mathcal{P}_0(\Gamma)$, $A \in \mathcal{A}_X$, 
and take any $\Lambda \in \mathcal{P}_0(\Gamma)$ with $X \subset \Lambda$. 
In this case, there $C>0$ for which 
\begin{equation} \label{power_law_lrb_2}
 \left\|T^{\Lambda}_t(A)-T^{X(r)}_t(A) \right\|  \leq  C \| A \| |X| \| \mathscr{L} \|_F \cdot \frac{t}{(1+ r)^{(1- \delta) \alpha_{\epsilon} - \nu}}
\end{equation}
whenever $r \geq 1$ and $0 \leq evt \leq r^{\delta}$. Moreover, one may take
\begin{equation} \label{constant_2} 
C =  \frac{\kappa C_{\epsilon}}{C_F}  \left(  \kappa 2^{2\alpha_{\epsilon}} 2^{1- \delta} \left( C_{\epsilon} + C_F \right) + e(C_F + \| F \|) \right) \, .
\end{equation}
    
    \end{thm}
    \begin{proof}
For any $R>0$, one has that
\begin{eqnarray} \label{3_term_bd}
 \left\|T^{\Lambda}_t(A)-T^{X(r)}_t(A) \right\| & \leq & \left\|T^{\Lambda}_t(A)-T^{\Lambda, R}_t(A) \right\| +  \left\| T^{\Lambda, R}_t(A)-T^{X(r), R}_t(A) \right\| +  \nonumber \\
& \mbox{ } & \quad +  \left\|T^{X(r), R}_t(A) -T^{X(r)}_t(A) \right\| \, .
\end{eqnarray}
An application of Theorem~\ref{thm:dyn_diff} shows that
\begin{eqnarray}
\left\|T^{\Lambda}_t(A)-T^{\Lambda, R}_t(A) \right\| & \leq & \|A\| \| \mathscr{L} \|_F G(R/2) \left( t |X(r)| + \frac{\E_{vt}(1+ rR^{-1})}{v C_F} \,
        \sum_{x \in X} \sum_{y \in \Lambda \setminus X(r)} F(d(x,y)) \right)  \\
        & \leq & \| A \| \| \mathscr{L} \|_F \left( \frac{ \kappa C_{\epsilon}}{(1+ \lfloor R/2 \rfloor)^{\alpha_{\epsilon}}} \cdot t \cdot |X| \kappa r^{\nu} + 
        \frac{ \kappa C_{\epsilon}}{(1+ \lfloor R/2 \rfloor)^{\alpha_{\epsilon}}}  \cdot \frac{e^{-r/R} t e^{evt}}{C_F} \cdot \frac{ \kappa C_{\epsilon}|X|}{(1+ r)^{\alpha_{\epsilon}}}\right) \nonumber
\end{eqnarray}
where we have also used bounds as in (\ref{term2_bd}) and (\ref{term3_bd}) from Theorem~\ref{thm:poly_dec_lrb}. 
Note that an identical estimate holds for the third term on the right-hand-side of (\ref{3_term_bd}) above. The time-independent 
exponential factor above motivates the choice of parametrization: $R = r^{1- \delta}$. We now bound the middle term on the
right-hand-side of (\ref{3_term_bd}). Arguing as in the proof of Theorem~\ref{thm:gen_sl_app}, we find that
\begin{equation} \label{middle_diff}
\left\| T^{\Lambda, R}_t(A)-T^{X(r), R}_t(A) \right\| \leq \sum_{\substack{Z \in S^{\Lambda}_{X(r)}: \\ {\rm diam}(Z) \leq R}} \int_0^t \| L_Z T_s^{X(r), R}(A) \| \, ds 
\end{equation}
For $r \geq 1$, it is clear that $R<r$ with our choice of parametrization. As such, for any $Z \in S^{\Lambda}_{X(r)}$ with ${\rm diam}(Z) \leq R$, we have that
$d(X,Z) > r-R>0$. In this case, we may apply Theorem~\ref{thm:appLRB} to all terms on the right-hand-side of (\ref{middle_diff}). We find that
\begin{eqnarray}
\left\| T^{\Lambda, R}_t(A)-T^{X(r), R}_t(A) \right\| & \leq &
\sum_{\substack{Z \in S^{\Lambda}_{X(r)}: \\ {\rm diam}(Z) \leq R}} \int_0^t  \frac{\| A \| \| L_Z \|_{cb} }{C_F}  \E_{vs}\left( \frac{d(X,Z)}{R} \right)  \sum_{x \in X}  \sum_{z \in Z} F(d(x,z) \, ds \nonumber \\
& \leq &
\frac{\|A \|}{C_F} \int_0^t \E_{vs} \left( \frac{r-R}{R} \right) \, ds \sum_{x \in X} \sum_{\substack{Z \in S^{\Lambda}_{X(r)}: \\ {\rm diam}(Z) \leq R}} \| L_Z \|_{cb}  \sum_{z \in Z} F(d(x,z)  
\nonumber \\
& \leq &
\frac{\|A \| (C_F + \| F \|)  }{C_F^2} \E_{vt} \left( \frac{r}{R} \right)  \sum_{x \in X} \sum_{y \in \Lambda \setminus X(r)} F(d(x,y))  \nonumber \\
& \leq & \frac{\|A \| \| \mathscr{L} \|_F(C_F + \| F \|)  }{C_F} \cdot e^{- r/R} t e^{evt+1} \cdot \frac{\kappa C_{\epsilon} |X|}{(1+r)^{\alpha_{\epsilon}}}
\end{eqnarray}
where we have used Lemma~\ref{lem:NOS}, (\ref{E_bd}), and that $v= \| \mathscr{L} \|_F C_F$. Setting $M = \kappa C_{\epsilon} \| \mathscr{L} \|_F \| A \| |X|$, we have found that 
\begin{eqnarray}
 \left\|T^{\Lambda}_t(A)-T^{X(r)}_t(A) \right\| & \leq & M \left( \frac{2 \kappa t r^{\nu}}{(1+ \lfloor R/2 \rfloor)^{\alpha_{\epsilon}}} + 
 \frac{2 \kappa C_{\epsilon} t e^{-r/R} e^{evt}}{C_F(1+ \lfloor R/2 \rfloor)^{\alpha_{\epsilon}} (1+r)^{\alpha_{\epsilon}}}  + \frac{(C_F + \| F \|) t e^{-r/R} e^{evt+1}}{C_F(1+r)^{\alpha_{\epsilon}}} \right)  \nonumber \\
 & \leq & M t \left( \frac{ \kappa 2^{2 \alpha_{\epsilon}} 2^{2-\delta}}{(1+ r)^{(1- \delta) \alpha_{\epsilon} - \nu} } +  \frac{\kappa C_{\epsilon} 2^{2 \alpha_{\epsilon}} 2^{2-\delta} e^{-r^{\delta}} e^{evt}}{C_F(1+ r)^{(2- \delta)\alpha_{\epsilon}} }  + \frac{e(C_F + \| F \|) e^{-r^{\delta}} e^{evt}}{C_F(1+r)^{\alpha_{\epsilon}}} \right) 
\end{eqnarray}
and moreover, the power are ordered by $(2 - \delta) \alpha_{\epsilon} > \alpha_{\epsilon} > (1- \delta) \alpha_{\epsilon} - \nu$. 
We have established (\ref{power_law_lrb_2}), whenever $0 \leq evt \leq r^{\delta}$, with $C$ as in (\ref{constant_2}). 
\end{proof}

%

%

\section{Correlation Decay}\label{sec:prop_of_corr}
In this section, we use the quasi-locality estimates proven in the previous section to obtain 
bounds on the decay of correlations in various states. First, we consider a class of states with 
quantified spatial decay, see Definition~\ref{def:spatiallydecayingstate}, and prove an estimate on dynamic correlations. A general statement of this
bound is the content of Theorem~\ref{thm:g_decaying}, and the methods we use are similar to those found in \cite{NachtergaeleOgataSims}. 
We supplement Theorem~\ref{thm:g_decaying} with more explicit estimates, depending on the decay of the
interaction generating the dynamics, in Lemma~\ref{lem:gen_cor_dec} and Lemma~\ref{lem:poly_cor_dec}. Next, we consider a class of dynamical fixed-points, see Definition~\ref{def:steadystate}. 
For these states, we obtain estimates on decay of correlations, see Theorem~\ref{thm:steadystate1}, using ideas which go back to \cite{Poulin}.    
Again, we supplement this general bound with more explicit estimates in Lemma~\ref{lem:gen_fp_dec} and Lemma~\ref{lem:poly_cor_dec_2}.

\subsection{Spatially Decaying States}\label{sec:spatiallydecaying}
The goal of this section is to estimate the effect of an irreversible dynamics 
on the decay of correlations for states with quantified spatial decay. 
To make this precise, we introduce the notion of a decay function which governs
the correlations for states on a local algebra. As in the previous section,
we consider quantum spin systems defined over a countable metric space $(\Gamma, d)$. 

\begin{define}\label{def:spatiallydecayingstate}
Let $\Lambda \subset \Gamma$ be finite and let $\omega$ be a state on the local algebra $\A_{\Lambda}$.
We say that $\omega$ has spatial correlations governed by $G_{\Lambda}$ if there is a function 
$G_{\Lambda}:[0, \infty) \times [0, \infty) \times [0, \infty) \to [0,2]$ for which given any $x,y \in [0, \infty)$ the map $d \mapsto G_{\Lambda}(x,y,d)$ is 
non-increasing and 
 \begin{equation}
            |\omega(AB) - \omega(A)\omega(B)| \le \| A \| \| B \| G_{\Lambda}(|X|,|Y|,d(X,Y))
\end{equation} 
for all $A \in \A_X$ and $B \in \A_Y$ where $X$ and $Y$ are arbitrary subsets of $\Lambda$.
\end{define}

Of course, for any state $\omega$ on $\A_{\Lambda}$, $\omega$ has correlations governed by 
the constant function $G_{\Lambda}(x,y,d) =2$. On the other hand, if $\omega$ is a 
product state, then one may take $G_{\Lambda}(x,y,d) =0$ for any $d>0$. In more interesting situations, $\omega$ has
correlations governed by a function of the form $G_{\Lambda}(x,y,d) = x^ny^m \tilde{G}_{\Lambda}(d)$ 
where $n$ and $m$ are non-negative integers and $\tilde{G}_{\Lambda}(d)$ decays swiftly to (or is identically) zero for large $d$. In any case, one has the following result.

\begin{thm}\label{thm:g_decaying}
Let $(\Gamma, d)$ be a countable metric space, $\mathscr{L} = \{L_Z \}_{Z \in \mathcal{P}_0(\Gamma)}$ a dissipative interaction,
and $\Lambda \subset \Gamma$ finite. If $\omega$ is a state on $\A_{\Lambda}$ with correlations
governed by $G_{\Lambda}$, then for any $A \in \A_X$ and $B \in \A_Y$ where $X,Y \subset \Lambda$ 
satisfy $d(X,Y) \geq 2$, one has that
\begin{equation} \label{base_cor_est}
\left|  \omega(T_t^{\Lambda}(AB)) - \omega(T_t^{\Lambda}(A)) \omega(T_t^{\Lambda}(B))\right| \leq \|A \| \|B \| G_{\Lambda}(|X(r)|, |Y(r)|,d(X(r), Y(r)) +C_{A,B}^{\Lambda}(r,t).
\end{equation}
for any $t \geq 0$ and $2 \leq 2r \leq d(X,Y)$. Here we have denoted by
\begin{eqnarray} \label{def_cab}
C_{A,B}^{\Lambda}(r,t) & = &  \| A \| \| T_t^{\Lambda}(B) - T_t^{Y(r)}(B) \| +  \| B \| \| T_t^{\Lambda}(A) - T_t^{X(r)}(A) \| \nonumber \\
& \mbox{ } & \quad + \| T_t^{\Lambda}(AB) - T_t^{(X \cup Y)(r)}(AB) \| \, .
\end{eqnarray} 
\end{thm}
\begin{proof} 
For any $t \geq 0$, set $B_t^{\omega} = B - \omega(T^{\Lambda}_t(B))\one$ and observe that
 \begin{equation}
 \omega(T_t^{\Lambda}(AB)) - \omega(T_t^{\Lambda}(A)) \omega(T_t^{\Lambda}(B)) = \omega(T_t^{\Lambda}(AB_t^{\omega})) \, .
 \end{equation}
 With $Z = X \cup Y$, one has that 
 \begin{equation} \label{comp_1}
 | \omega(T_t^{\Lambda}(AB_t^{\omega})) | \leq |\omega(T_t^{Z(r)}(AB_t^{\omega}))|  + |\omega(T_t^{\Lambda}(AB_t^{\omega})) - \omega(T_t^{Z(r)}(AB_t^{\omega})) |
 \end{equation}
 for any $r \geq 1$. If $2 \leq 2r < d(X,Y)$, then
 \begin{equation}
 T_t^{Z(r)}(AB_t^{\omega}) = T_t^{Z(r)}(AB) - \omega(T_t^{\Lambda}(B)) T_t^{Z(r)}(A) = T_t^{X(r)}(A) T_t^{Y(r)}(B) - \omega(T_t^{\Lambda}(B)) T_t^{X(r)}(A) 
 \end{equation} 
In this case, the first term on the right hand side of (\ref{comp_1}) can be estimated as
\begin{eqnarray} \label{1_term}
 |\omega(T_t^{Z(r)}(AB_t^{\omega}))| & \leq & \left| \omega(T_t^{X(r)}(A) T_t^{Y(r)}(B)) - \omega(T_t^{X(r)}(A)) \omega(T_t^{Y(r)}(B)) \right|   \nonumber \\
 & \mbox{ } & \quad +  \left| \omega(T_t^{X(r)}(A)) \left( \omega \left( T_t^{\Lambda}(B) - T_t^{Y(r)}(B) \right) \right) \right| \nonumber \\
 & \leq & \|A\| \| B \| G_{\Lambda}\left(|X(r)|, |Y(r)|, d(X(r), Y(r))\right) + \| A \| \| T_t^{\Lambda}(B) - T_t^{Y(r)}(B) \| 
\end{eqnarray}
For the second term on the right hand side of (\ref{comp_1}), note that
\begin{equation}
T_t^{\Lambda}(AB_t^{\omega}) - T_t^{Z(r)}(AB_t^{\omega}) = T_t^{\Lambda}(AB) - T_t^{Z(r)}(AB) - \omega(T_t^{\Lambda}(B)) \left( T_t^{\Lambda}(A) - T_t^{X(r)}(A) \right) 
\end{equation}
and therefore, 
\begin{equation} \label{2_term}
|\omega(T_t^{\Lambda}(AB_t^{\omega})) - \omega(T_t^{Z(r)}(AB_t^{\omega})) | \leq \| T_t^{\Lambda}(AB) - T_t^{Z(r)}(AB) \| + \| B \| \| T_t^{\Lambda}(A) - T_t^{X(r)}(A) \| 
\end{equation}
Combining (\ref{comp_1}), (\ref{1_term}), and (\ref{2_term}), we have established (\ref{base_cor_est}). 
\end{proof}

Generally, the quantity $C_{A,B}^{\Lambda}(r,t)$ estimates the correlations which 
this irreversible dynamics generates between $A$ and $B$. In principle, the estimate above holds for all $r \geq 1$, however,
once $d(X(r), Y(r))=0$, $\omega$ no longer provides any spatial decay. Note that if $\omega$ is a product state, then Theorem~\ref{thm:g_decaying} implies that 
\begin{equation}
|  \omega(T_t^{\Lambda}(AB)) - \omega(T_t^{\Lambda}(A)) \omega(T_t^{\Lambda}(B)) |  \leq C_{A,B}^{\Lambda}(r, t) 
\end{equation}
for all $t \geq 0$ and $2 \leq 2r < d(X,Y)$. 

The quasi-locality estimates we developed in the previous section allow us to provide estimates on this quantity $C_{A,B}^{\Lambda}(r,t)$. 
We collect this information in the two lemmas below.  

\begin{lem} \label{lem:gen_cor_dec} Let $(\Gamma, d)$ be a countable metric space equipped with an $F$-function $F$ and $\mathscr{L}= \{ L_Z \}_{Z \in \mathcal{P}_0(\Gamma)}$ a dissipative interaction with $\| \mathscr{L} \|_F< \infty$.
Let $\Lambda \subset \Gamma$ be finite, take $X,Y \subset \Lambda$ with $1 \leq d(X,Y)$, and consider $A \in \A_X$ and $B\in \A_Y$. One has that
\begin{eqnarray} \label{gen_cor_est}
C_{A,B}^{\Lambda}(r,t) & \leq & 2 \|A\| \| B \| \| \mathscr{L}\|_{F}  \left(t + \frac{C_{F}  +   \| F \|}{C_{F}} \int_0^t (e^{vs} -1) \, ds \right)  \times \nonumber \\
& \mbox{ } & \quad  \times \left(  \sum_{x \in X}   \sum_{y \in \Lambda \setminus X(r)} F(d(x,y)) +
 \sum_{y \in Y} \sum_{x \in \Lambda \setminus Y(r)} F(d(x,y)) \right) 
\end{eqnarray}
for all $r \geq 1$.
\end{lem}

\begin{proof} Recall that
\begin{eqnarray} \label{recall_def_cab}
C_{A,B}^{\Lambda}(r,t) & = &  \| A \| \| T_t^{\Lambda}(B) - T_t^{Y(r)}(B) \| +  \| B \| \| T_t^{\Lambda}(A) - T_t^{X(r)}(A) \| \nonumber \\
& \mbox{ } & \quad + \| T_t^{\Lambda}(AB) - T_t^{(X \cup Y)(r)}(AB) \| \, .
\end{eqnarray} 
To prove (\ref{gen_cor_est}), we apply Theorem~\ref{thm:gen_sl_app}. In fact, a direct application of Theorem~\ref{thm:gen_sl_app} shows that
\begin{equation} \label{cor_est_A}
 \| B \| \| T_t^{\Lambda}(A) - T_t^{X(r)}(A) \|  \leq  \|A\| \| B \| \| \mathscr{L}\|_{F}  \left(t + \frac{C_{F}  +   \| F \|}{C_{F}} \int_0^t (e^{vs} -1) \, ds \right) 
 \sum_{x \in X} \sum_{y \in \Lambda \setminus X(r)} F(d(x,y)) 
\end{equation}
and an analogous bound estimates $ \| A \| \| T_t^{\Lambda}(B) - T_t^{Y(r)}(B) \|$. For the final term in (\ref{recall_def_cab}), we argue similarly and use that
\begin{equation}
\sum_{w \in X \cup Y} \sum_{z \in \Lambda \setminus (X \cup Y)(r)} F(d(w,z)) \leq 
\sum_{x \in X}   \sum_{y \in \Lambda \setminus X(r)} F(d(x,y)) +
 \sum_{y \in Y} \sum_{x \in \Lambda \setminus Y(r)} F(d(x,y))
\end{equation}
\end{proof}

\begin{lem}  \label{lem:poly_cor_dec} Let $(\Gamma, d)$ be $\nu$-regular and $\mathscr{L} = \{L_Z \}_{Z \in \mathcal{P}_0(\Gamma)}$ be a dissipative interaction with
$\| \mathscr{L} \|_F< \infty$ for the choice of $F$-function
\begin{equation} \label{power_law_F_2}
F(x) = \frac{1}{(1+x)^{\alpha}} \quad \mbox{for all } x \geq 0
\end{equation}
with some $\alpha > 2\nu +1$. Fix $\epsilon \in (0, \alpha - 2 \nu -1)$, set $\alpha_{\epsilon} = \alpha - \nu - 1- \epsilon$, and 
let $0 < \delta < 1$ satisfy $(1- \delta)\alpha_{\epsilon} > \nu$. 

Let $\Lambda \subset \Gamma$ be finite, take $X,Y \subset \Lambda$ with $1 \leq d(X,Y)$, and consider $A \in \A_X$ and $B\in \A_Y$. One has that
\begin{equation}
C_{A,B}^{\Lambda}(r,t) \leq  3 C \| A \| \| B \| \left( |X| + |Y| \right) \| \mathscr{L} \|_F \cdot \frac{t}{(1+ r)^{(1- \delta) \alpha_{\epsilon} - \nu}}
\end{equation}
whenever $r \geq 1$ and $0 \leq evt \leq r^{\delta}$. Here $C$ is as in (\ref{constant_2}). 
\end{lem}
\begin{proof}
One readily checks that the proof of this lemma is a direct consequence of Theorem~\ref{thm:sl_app_poly_dec} and simple over-estiimates on the
support of the observables. 
\end{proof}

\subsection{Dynamical Fixed-Points}\label{sec:dynamicalsteadystates}

The goal of this section is to quantify decay of correlations for a class of
dynamical fixed points. We discuss our assumptions on convergence to the fixed points
in relations to other applications in the literature in an appendix. As in the previous section, we first prove
a general bound, see Theorem~\ref{thm:steadystate1} below, and then follow-up with more explicit estimates in 
Lemma~\ref{lem:gen_fp_dec} and Lemma~\ref{lem:poly_cor_dec_2}. 

\begin{define}\label{def:steadystate}
Let $(\Gamma, d)$ be a countable metric space and $\mathscr{L}= \{ L_Z \}_{Z \in \mathcal{P}_0(\Gamma)}$ be a dissipative interaction.
Take $\Lambda \subset \Gamma$ finite and consider $\pi$ a state on the local algebra $\A_{\Lambda}$.
We say that $\pi$ is a local dynamical fixed-point of $\mathscr{L}$ with convergence governed by $g_{\Lambda}$ if
        \begin{enumerate}
            \item $\pi \circ T^{\Lambda}_t = \pi$ for all $t\ge 0$. 
            \item There exists a function $g_\Lambda:[0,\infty) \to [0,2]$ for which one has 
            \[ | (\pi - \omega)\circ T^{\Lambda}_t(A) | \le g_\Lambda(t) \| A\| \] 
            for all states $\omega$ on $\A_{\Lambda}$, $A \in \A_{\Lambda}$, and $t \geq 0$.
        \end{enumerate}  
\end{define}  

For any local dynamical fixed-point, \textit{i.e.} state $\pi$ on $\A_{\Lambda}$ satisfying 1 above, the function
$g_{\Lambda}(t) =2$ governs convergence to $\pi$ in a trivial sense. We are more interested
in situations where the convergence is governed by a non-increasing function $g_{\Lambda}$ with
$g_{\Lambda}(t) \to 0$ as $t \to \infty$. In the later case, one readily checks that any such local dynamical fixed-point is
unique. 

The following result estimates correlations in such local dynamical fixed-points.

\begin{thm}\label{thm:steadystate1}
Let $(\Gamma, d)$ be a countable metric space, $\mathscr{L} = \{L_Z \}_{Z \in \mathcal{P}_0(\Gamma)}$ a dissipative interaction,
and $\Lambda \subset \Gamma$ finite. Let $\pi$ be a local dynamical fixed-point for $\mathscr{L}$ with 
convergence governed by $g_{\Lambda}$. For any $A \in \A_X$ and $B \in \A_Y$ where $X,Y \subset \Lambda$ 
satisfy $d(X,Y) >0$, one has that
\begin{equation} \label{gen_fp_dec}
            |\pi(AB) - \pi(A) \pi(B)| \leq  | \omega(T_t^{\Lambda}(AB_t^{\omega}))| + 3 \| A \| \| B \| g_{\Lambda}(t) 
\end{equation}
for any state $\omega$ on $\A_{\Lambda}$ and $t \geq 0$.
\end{thm}
\begin{proof}
As in the proof of Theorem~\ref{thm:g_decaying}, set $B^\pi = B- \pi(B)\one$. Clearly, $\|B^\pi\|\le 2 \|B\|$ and moreover,
\begin{equation}
\pi(AB) - \pi(A) \pi(B) = \pi(AB^\pi).
\end{equation}
Since $\pi$ is time-invariant, we have that for any $t \geq 0$ 
\begin{eqnarray} \label{comp_2}
| \pi(AB^\pi)|  = | \pi(T_t^{\Lambda}(AB^\pi))| & \leq & | \omega(T_t^{\Lambda}(AB^\pi))| + | (\pi - \omega)(T_t^{\Lambda}(AB^\pi))| \nonumber \\
& \leq &  | \omega(T_t^{\Lambda}(AB^\pi)| + 2 g_{\Lambda}(t) \| A \| \| B \|  
\end{eqnarray}
and the above holds for any state $\omega$ on $\A_{\Lambda}$.
Rewriting things, one finds that 
\begin{equation}
B^\pi = B - \pi(B) \one = B - \pi(T_t^{\Lambda}(B)) \one = B_t^{\omega} - (\pi - \omega )(T_t^{\Lambda}(B)) \one
\end{equation}
where we have, again, used notation as in the proof of Theorem~\ref{thm:g_decaying}. 
In this case, the first term on the right hand side of (\ref{comp_2}) may be estimated as
\begin{eqnarray} \label{cor_est}
 | \omega(T_t^{\Lambda}(AB^\pi))| & \leq &  | \omega(T_t^{\Lambda}(AB_t^{\omega}))| + |(\pi - \omega)(T_t^{\Lambda}(B))| | \omega(T_t^{\Lambda}(A))|  \nonumber \\
 & \leq & | \omega(T_t^{\Lambda}(AB_t^{\omega}))| +g_{\Lambda}(t) \| A \| \| B \|
\end{eqnarray}
Combining (\ref{comp_2}) and (\ref{cor_est}), we have proven (\ref{gen_fp_dec}). 
\end{proof}

 Depending on the decay of the interactions, one may be more explicit with these bounds.
 In fact, if the interactions decay exponentially, we have the following.

\begin{lem} \label{lem:gen_fp_dec} Let $(\Gamma, d)$ be a countable metric space equipped with an $F$-function $F_0$ and $\mathscr{L}= \{ L_Z \}_{Z \in \mathcal{P}_0(\Gamma)}$ a dissipative interaction with $\| \mathscr{L} \|_{F_a}< \infty$ for some $a>0$ where $F_a$ is the $F$-function on $(\Gamma, d)$ given by $F_a(r) = e^{-ar}F_0(r)$.
Let $\Lambda \subset \Gamma$ be finite and $\pi$ be a local dynamical fixed-point for $\mathscr{L}$ with 
convergence governed by $g_{\Lambda}$. For any $X,Y \subset \Lambda$ with $d(X,Y)>2$, consider $A \in \A_X$ and $B\in \A_Y$. One has that
\begin{eqnarray} \label{fp_cor_dec}
|\pi(AB) - \pi(A) \pi(B)| & \le & 2 \|A\| \|B\| (|X| + |Y|) \| \mathscr{L} \|_{F_a} \| F_0 \|  \left(t_a + \frac{C_{F_a}  +   \| F_a \|}{C_{F_a}} \int_0^{t_a} (e^{v_as} -1) \, ds \right) e^{-\frac{ad(X,Y)}{2}} \nonumber \\
& \mbox{ } & \quad + 3 \| A \| \|B \| g_{\Lambda}(t_a) 
\end{eqnarray} 
where $v_a = \| \mathscr{L} \|_{F_a}C_{F_a}$ and $t_a = \frac{ad(X,Y)}{4v_a}$. 
\end{lem}
Before we prove this bound, note that a simple estimate shows that
\begin{equation}
 \left(t_a + \frac{C_{F_a}  +   \| F_a \|}{C_{F_a}} \int_0^{t_a} (e^{v_as} -1) \, ds \right) e^{-\frac{ad(X,Y)}{2}}  \leq  \left( \frac{2C_{F_a}+\|F_a \|}{v_aC_{F_a}} \right) \frac{ad(X,Y)}{4}e^{-\frac{ad(X,Y)}{4}}
\end{equation}
and so the first term above in (\ref{fp_cor_dec}) decays exponentially in $d(X,Y)$. 

\begin{proof}
An application of Theorem~\ref{thm:steadystate1} shows that
\begin{equation}
  |\pi(AB) - \pi(A) \pi(B)| \leq  | \omega(T_t^{\Lambda}(AB_t^{\omega}))| + 3 \| A \| \| B \| g_{\Lambda}(t) 
  \end{equation}
 for any state $\omega$ and $t \geq 0$. Taking $\omega$ to be a product state, we find that
 \begin{equation} \label{base_cor_est}
 | \omega(T_t^{\Lambda}(AB_t^{\omega}))| = |\omega(T^{\Lambda}_t(AB) - \omega(T^{\Lambda}_t(A))\omega(T^{\Lambda}_t(B))|\le C_{A,B}^{\Lambda}(r,t) 
 \end{equation}
 for all $t \geq 0$ and $2 \leq 2r <d(X,Y)$, where we recall that $C_{A,B}^{\Lambda}(r,t)$ is as defined in (\ref{def_cab}). Lemma~\ref{lem:gen_cor_dec} provides an estimate
 for $C_{A,B}^{\Lambda}(r,t)$ which is relevant in this case:
 \begin{eqnarray} \label{cor_bd_1}
C_{A,B}^{\Lambda}(r,t) & \leq & 2 \|A\| \| B \| \| \mathscr{L}\|_{F_a}  \left(t + \frac{C_{F_a}  +   \| F_a \|}{C_{F_a}} \int_0^t (e^{v_as} -1) \, ds \right)  \times \nonumber \\
& \mbox{ } & \quad  \times \left(  \sum_{x \in X}   \sum_{y \in \Lambda \setminus X(r)} F_a(d(x,y)) +
 \sum_{y \in Y} \sum_{x \in \Lambda \setminus Y(r)} F_a(d(x,y)) \right) \nonumber \\
& \leq & 2 \|A\| \| B \| \| \mathscr{L}\|_{F_a}  \left(t + \frac{C_{F_a}  +   \| F_a \|}{C_{F_a}} \int_0^t (e^{v_as} -1) \, ds \right) e^{-ar} \| F_0 \| (|X|+|Y|) 
\end{eqnarray}
As (\ref{cor_bd_1}) holds for all  $2 \leq 2r < d(X,Y)$, the estimate extends to $r=d(X,Y)/2$.
The bound claimed in (\ref{fp_cor_dec}) now follows
by choosing $t = \frac{ad(X,Y)}{4v_a}$.
\end{proof}

An analogous bound holds for polynomially decaying interactions.

\begin{lem}  \label{lem:poly_cor_dec_2} Let $(\Gamma, d)$ be $\nu$-regular and $\mathscr{L} = \{L_Z \}_{Z \in \mathcal{P}_0(\Gamma)}$ be a dissipative interaction with
$\| \mathscr{L} \|_F< \infty$ for the choice of $F$-function
\begin{equation} \label{power_law_F_2}
F(x) = \frac{1}{(1+x)^{\alpha}} \quad \mbox{for all } x \geq 0
\end{equation}
with some $\alpha > 2\nu +1$. Fix $\epsilon \in (0, \alpha - 2 \nu -1)$, set $\alpha_{\epsilon} = \alpha - \nu - 1- \epsilon$, and 
let $0 < \delta < 1$ satisfy $(1- \delta)\alpha_{\epsilon} > \nu$. 

Let $\Lambda \subset \Gamma$ be finite and $\pi$ be a local dynamical fixed-point for $\mathscr{L}$ with 
convergence governed by $g_{\Lambda}$. For any $X,Y \subset \Lambda$ with $d(X,Y)>2$, consider $A \in \A_X$ and $B\in \A_Y$. 
For any $0< \eta < \min[ \delta, (1-\delta) \alpha_{\epsilon} - \nu]$, one has that
\begin{eqnarray}\label{eqn:dss_power_law}
|\pi(AB) - \pi(A) \pi(B)| & \le &  C'\|A\| \|B\| (|X|+|Y|) \| \mathscr{L} \|_F \cdot \frac{1}{(1+d(X,Y))^{(1-\delta)\alpha_\epsilon - \nu-\eta}}  \nonumber \\
& \mbox{ } & \quad + 3 \| A \| \| B \| g_{\Lambda} \left( \frac{d(X,Y)^{\eta}}{ev2^{\eta}} \right) 
\end{eqnarray}
where 
\begin{equation}
C' = \frac{3C}{ev} \cdot 2^{(1- \delta) \alpha_{\epsilon} - \nu - \eta}
\end{equation}
and $C$ is as in (\ref{constant_2}).
\end{lem}

\begin{proof}
Arguing as in the proof of Lemma~\ref{lem:gen_fp_dec}, using a product state in the application of Theorem~\ref{thm:steadystate1} shows 
\begin{equation} \label{fp_base_est_2}
  |\pi(AB) - \pi(A) \pi(B)| \leq  C_{A,B}^{\Lambda}(r,t) + 3 \| A \| \| B \| g_{\Lambda}(t) 
 \end{equation}
for any $t \geq 0$ and $2\leq 2r<d(X,Y)$. We need only estimate the correlations, \textit{i.e.} $C_{A,B}^{\Lambda}(r,t)$. 

In this case, Lemma~\ref{lem:poly_cor_dec} generally implies that
\begin{equation} \label{cor_est_2}
C_{A,B}^{\Lambda}(r,t) \leq  3 C \| A \| \| B \| \left( |X| + |Y| \right) \| \mathscr{L} \|_F \cdot \frac{t}{(1+ r)^{(1- \delta) \alpha_{\epsilon} - \nu}}
\end{equation}
whenever $r \geq 1$ and $0 \leq evt \leq r^{\delta}$. For any $0< \eta < \min[ \delta, (1- \delta) \alpha_{\epsilon} - \nu]$ and choice of $t$ satisfying
\begin{equation} \label{def_t}
evt = \left( \frac{d(X,Y)}{2} \right)^{\eta}
\end{equation}
the bound in (\ref{cor_est_2}) holds for all $r< d(X,Y)/2$ satisfying $0 \leq evt \leq r^{\delta}$. One readily checks that such values of $r$
exist, and therefore, the right-hand-side of (\ref{fp_base_est_2}) may be estimated by the right-hand-side of (\ref{cor_est_2}) evaluated with $t$ as in (\ref{def_t}) and
$r=d(X,Y)/2$. The bound claimed in (\ref{eqn:dss_power_law}) follows.
\end{proof}

\appendix
\renewcommand{\thesection}{\Alph{section}}
\renewcommand{\thesubsection}{A.\Roman{subsection}}

\section{Appendix: Some Semigroup Theory}
This appendix seeks to show that our definition of a local dynamical fixed-point as in Definiton~\ref{def:steadystate} is logically equivalent to a few other conditions that are present in the literature. In particular, we show that  $\|(\pi - \psi)(T^{\Lambda}_t(A))| \to 0$ uniformly exponentially fast (see Theorem~\ref{thm:exponential_steady_state}). We also show that our notion is equivalent to the assumption that the associated Schr\"odinger picture semigroup satisfies rapid mixing in the sense of \cite[Theorem 1]{KastoryanoEisert}  (see Theorem~\ref{thm:rapid_mixing}). Because our constants depend on the volume, existence of a local dynamical fixed point implies neither global nor local rapid mixing in the sense of \cite{Brandao, Cubitt}. 

We prove the following general theorem 
\begin{thm}\label{thm:exponential_steady_state}
	Let $\mathcal{H}$ be a finite-dimensional Hilbert space and let $L:B(\mathcal{H})\to B(\mathcal{H})$ be a bounded Lindblad generator with associated quantum dyanmical semigroup $(T_t | t\ge 0)$ in the Heisenberg picture. By $(T'_t|t\ge 0)$ we mean the semigroup acting on $B(\mathcal{H})^*$ by pre-composition. The following are equivalent: 
	\begin{enumerate}
		\item There exists a state $\pi\in B(\mathcal{H})^*$ so that $\pi\circ T_t = \pi$ for all $t\ge 0$ and there exists a function $g_{\pi}:[0, \infty) \to [0, \infty)$ so that $g_{\pi}(t) \to 0$ as $t\to \infty$ and for all  states $\psi\in B(\mathcal{H})^*$, and all $A\in B(\mathcal{H})$, one has 
			\begin{equation}
				|(\psi - \pi)\circ T_t(A)| \le \|A\| g_{\pi}(t)\,.
			\end{equation} We call such a state a \textit{dynamical fixed point} of $(T_t|t\ge 0)$. 
		\item There exist constants $c>0$ and $\gamma >0$, and a rank-one projection $P:B(\mathcal{H})^* \to B(\mathcal{H})^*$ so that 
		\begin{equation}
			\|T'_t - P\|\le ce^{-\gamma t}\,.
		\end{equation} 
	\end{enumerate}
\end{thm}

Before beginning the proof, let us mention an immediate application to our models in the following: 
\begin{cor}
	Let $(\Gamma, d)$ be a discrete countable metric space, and let $\mathscr{L}$ be a local dissipative interaction. Fix $\Lambda \in \mathcal{P}_0(\Gamma)$ and let $(T^{\Lambda}_t|t\ge0)$ be the associated dynamical semigroup. Then, the following are equivalent:
	\begin{enumerate}
		\item $\pi$ is a local dynamical fixed point (see Definition~\ref{def:steadystate}). 
		\item The state $\pi$ is a fixed point of $(T^{\Lambda}_t|t\ge0)$ and there exist constants $c,\gamma>0$ which may depend on $\Lambda$ so that 
	\begin{equation}
		|(\pi- \psi)\circ T^{\Lambda}_t(A)| \le c \|A\| e^{-\gamma t}\,, \quad \forall \psi \in \A_\Lambda^* \text{ and } A\in \A_{\Lambda}\,.
	\end{equation} for all states $\psi$ and all $A\in \A_{\Lambda}$. 
	\end{enumerate}
	In particular, if one of the above equivalent conditions holds in $\Lambda$, one can take $g_\Lambda(t) = ce^{-\gamma t}$.
\end{cor}
\begin{proof}
By Theorem~\ref{thm:exponential_steady_state}, there is a rank-one projection $P:\A_{\Lambda}^* \to \A_{\Lambda}^*$ so that the pre-composition semigroup $S_t:= (T^{\Lambda}_t)'$ converges to $P$ exponentially fast in the norm of $\A_{\Lambda}^*$. By the proof of Theorem~\ref{thm:exponential_steady_state} below, it follows that $P(\pi) = \pi$ and for all $f\in \A_{\Lambda}^*$, one has $P(f) = f(\one) \pi$. Then, 
	\begin{equation}
		|(\psi - \pi)\circ T^{\Lambda}_t(A)| \le \|A\|  \|S_t(\psi) - S_t(\pi)\| = \|A\| \|(S_t-P)(\psi)\| \le  \|A\| \|S_t - P\| \le c \|A\| e^{-\gamma t} \,.
	\end{equation} The reverse implication is obvious.
\end{proof}

Before we begin the proof of Theorem~\ref{thm:exponential_steady_state}, we will recall several lemmas for the convenience of the reader. These can be found in more generality in \cite{EngelNagel_Short}. We shall fix $X$ as a finite-dimensional vector space equipped with a norm $\| \cdot \|$. Moreover, $(S_t:t\ge 0)$ will denote a semigroup of linear operators $S_t:X\to X$ so that $S_t$ is uniformly continuous in the sense that $\|S_t - \id\| \to 0$ as $t\downarrow 0$. We will write $A$ for the generator of $S_t$, which we note is a bounded operator on $X$.

\begin{lem}[Proposition II.2.3 \cite{EngelNagel_Short}]\label{lem:subspace_semigroup}
	Let $X$ be a vector space and let $Y\subset X$ be a vector subspace of $X$. Let $(S_t|t\ge 0)$ be a norm continuous semigroup on $X$ and suppose that $S_t(Y)\subset Y$ for all $t\ge 0$. If $A$ is the generator of $S_t$, then the restriction $S_t|_Y$ is a norm-continuous semigroup with generator $A|_Y$. 
\end{lem}
\begin{proof}
	The continuity and semigroup parts of the claim are obvious. To see that $A|_Y$ generates $S_t|_Y$, note that for any $y\in Y$, we have 
	\begin{equation}
		\left \|\frac{S_t|_Y y - y}{t} - A|_Y y\right \|_Y = \left\|\frac{S_ty - y }{t}- Ay\right\|_X \to 0 \quad \text{ as } t\downarrow 0.
	\end{equation}  Since the generator of a semigroup is unique, it must be that $A|_Y$ generates the restricted semigroup. 
\end{proof}

Let us recall from \cite[Definition I.1.5]{EngelNagel_Short} the \textit{growth bound} $\omega_0$ of a semigroup $(S_t : t\ge 0)$ is the quantity \begin{equation} \omega_0 = \inf \{r \in \mathbb{R}: \exists M_r\ge 1 \colon \|S_t\|\le M_re^{rt}\}\,.\end{equation} Recall that the spectral radius $rad(S_t)$ is the supremum over all the elements of the spectrum of $S_t$. 

Let us recall an elementary subadditivity result. 
\begin{lem}\label{lem:Fekete}
Let $f:(0, \infty) \to [0, \infty)$ be a measurable function that is subadditive and uniformly bounded on compact intervals. Then, the following limits exist:
	\begin{equation}
		\lim_{t\to \infty} \frac{f(t)}{t} = \inf_{t>0} \frac{f(t)}{t}\,.
	\end{equation}
\end{lem} For a proof, see \cite[Lemma V.1.23]{EngelNagel_Short}.

\begin{lem}[Proposition V.1.22 \cite{EngelNagel_Short}]\label{lem:appendix_growth_bound}
	Let $S_t$ be a norm-continuous semigroup on a vector space $X$ with generator $A$. The following equalities hold:
	\begin{equation}
		\omega_0 = \inf_{t>0} t^{-1} \log \|S_t\| = \lim_{t\to \infty} t^{-1} \log \|S_t\|
	\end{equation} In particular, have $rad(S_t) = e^{t \omega_0}$. 
\end{lem}
\begin{proof}
	Note that the function $f(t) := \log \|S_t\|$ is a measurable subadditive function. Moreover, since $A$ is bounded, we have 
	\begin{equation}
		\log(\|S_t\|) \le \log e^{t\|A\|} = t\|A\|\,.
	\end{equation} Hence, $f(t)$ is uniformly bounded on compact intervals. Therefore, by the Lemma~\ref{lem:Fekete}, $v = \inf t^{-1}f(t) = \lim_{t\to \infty} t^{-1}f(t)$ exists. Note for all $t>0$, $e^{vt} \le e^{\frac{t}{t} \log \|S_t\|}= \|S_t\|$ by definition. One has $e^{v\cdot 0} = 1 = \|S_0\|$ as well since $S_0 = \id$. Thus $v\le \omega_0$. 
	
	On the other hand, if $\epsilon >0$, there is some $t_0\ge 0$ so that for any $t\ge t_0$, we have $t^{-1}f(t) \le v+\epsilon$. Hence $\|S_t\| \le e^{(v+\epsilon) t}$. Moreover, on $[0,t_0]$, we may bound $\|S_t\|$ above uniformly by $e^{t_0 \|A\|}\ge 1$. Thus, $\|S_t \| \le M_{v+\epsilon} e^{(v+\epsilon)t}$ for all $t\ge 0$ by taking $M_{v+\epsilon} = \max\{1, e^{t_0 \|A\|}\}$. Hence $\omega_0\le v+\epsilon$, and so it must be that $\omega_0 = v$. 
	
	Now, by the Gelfand-Hadamard formula (\textit{cf.} \cite{Bhatia} or \cite{Weidmann}) and the semigroup property, we know that for any $t>0$, 
		\begin{equation}
			rad(S_t) = \lim_{n\to \infty} \|T_{nt}\|^{1/n} = \lim_{n\to \infty} \exp \left( \frac{t}{nt} \log \|T_{nt}\| \right)= e^{t\omega_0}
		\end{equation} by continuity. That the relation holds when $t=0$ follows readily from the fact that $S_0 = \id$, completing the proof.
		\end{proof}
The contents of the following Lemma can be found in Propositions V.3.3.2 and V.3.3.5 of \cite{EngelNagel_Short}. See also Lyapunov's Theorem V.3.3.6 in \cite{EngelNagel_Short}.
\begin{lem}[Lyapunov's Theorem]
Let $(S_t|t\ge 0)$ be a semigroup on a vector space $X$. The following are equivalent: 
	\begin{enumerate}
		\item The limit $\lim_{t\to \infty} \|S_t\| = 0$. 
		\item There exists $\epsilon >0$ and $M>0$ so that 
			\begin{equation}
				\|S_t\| \le Me^{-\epsilon t}
			\end{equation} \textit{I.e.}, $S_t$ is \textit{uniformly exponentially stable}.
	\end{enumerate}
\end{lem}
\begin{proof} 
	That 2 implies 1 is obvious here. To see the other direction, recall that $e^{t\omega_0} = rad(S_t) \le \|S_t\|$. Hence $\omega_0<0$, and the inequality follows from Lemma~\ref{lem:appendix_growth_bound}.
\end{proof}

\begin{lem}\label{lem:exp_decay}
	Let $S_t$ be a norm-continuous semigroup on a finite dimensional vector space $X$ with a generator $A$. Assume that the limit
		\begin{equation}
			\|\cdot \|-\lim_{t\to \infty} S_t = P \neq 0
		\end{equation} exists. Then, there are constants $M\ge 1$ and $\epsilon >0$ so that 
		\begin{equation}
			\|S_t - P\| \le M e^{-\epsilon t}
		\end{equation}
\end{lem}

	We follow the general idea of the proof given in Lemma V.4.2 of \cite{EngelNagel_Short}. 
\begin{proof}
	First, we claim that $P$ is a projection onto the subspace $F:=\{x\in X \colon S_tx = x \, \forall t\ge 0\}$. To see this, note that for any $s\ge 0$, one has 
	\begin{align*}
		S_s P = \lim_{t\to \infty} S_s S_t = \lim_{t\to \infty} S_{t}S_s = \lim_{t\to \infty} S_{s+t} = P\,.
	\end{align*} Therefore, if $S_tx = x$ for all $t\ge 0$, then $Px = x$ by passing to the limit. Hence, $F\subset PX$. On the other hand, $S_tP = P$ so for any $y\in PX$, there is $x$ so that $y= Px$ and in particular $S_t y = S_t Px = Px = y$ for any $t\ge 0$. That is, $y$ is a fixed point of $S_t$. 
	
	Thus $X = PX \oplus (\id - P)X = F\oplus (\id-P)X$. In particular, the restriction $S_t|_{(\id-P)X}$ is given by the composition $S_t(\id - P)$. Moreover, by the commutation of $S_t$ and $P$, we infer that  $(S_t|_{(\id - P)X} | t\ge 0)$ is a norm continuous semigroup such that \[\lim_{t\to \infty} \|S_t(\id-P)\| = \lim_{t\to \infty} \|S_t - S_t P\| = \lim_{t\to \infty} \|S_t - P\| = 0\,.\] 
	
	In other words, $S_t(\id - P)$ is an exponentially stable semigroup on $(\id - P)X$. Therefore, by Lyapunov's theorem, there exists $M\ge 1$ and $\epsilon >0$ so that \begin{equation}
		\|S_t|_{(\id - P)X}\|=\|S_t(\id-P)\| = \|S_t - P\| \le Me^{-\epsilon t},
	\end{equation} as claimed.
\end{proof}

Recall that one may equip $B(\mathcal{H})^*$ with an involution given by $\xi^*(x) := \overline{\xi(x^*)}$ for all $x\in B(\mathcal{H})$ where $\xi\in B(\mathcal{H})^*$.  One then has that $\xi = \Re(\xi) + i \Im(\xi)$ where $2 \Re(\xi) = \xi + \xi^*$ and $2\Im(\xi) = \xi - \xi^*$.  Recall any linear functional such that $\xi^* = \xi$ is called \textit{Hermitian}.  The Jordan decomposition \cite[Proposition 3.2.7]{BratteliRobinson_1} (see \cite[Theorem 4.3.6]{KadisonRingrose_1} for a nice proof) asserts for any Hermitian linear functional $\xi_h=\xi_h^*$ there exist positive linear functionals $\xi_1, \xi_2$ so that $\xi_h = \xi_1 - \xi_2$ and $\|\xi_h\| = \|\xi_1\| + \|\xi_2\|$.  In particular, any linear functional $\xi = \xi_1 - \xi_2 +i\xi_3 - i \xi_4$, where the $\xi_j$ are positive linear functionals and $\|\Re(\xi)\| = \|\xi_1\|+ \|\xi_2\|$ and $\|\Im(\xi)\| = \|\xi_3\|+ \|\xi_4\|$.

We are now ready to complete the proof of Theorem~\ref{thm:exponential_steady_state}. 
\begin{proof}
	($1.\Rightarrow 2.$):  Let $P:B(\mathcal{H})^* \to B(\mathcal{H})^*$ via $P(\xi) = \xi(\one) \pi$.  Observe that $P$ is a bounded, idempotent operator with $P(\pi) = \pi$. Observe that for any nonzero positive linear functional $f\in B(\mathcal{H})^*$, we have the estimate 
	\[
		\|(f - f(\one)\pi)\circ T_t\| = f(\one) \|( \tilde f - \pi) \circ T_t\| \le f(\one) g_\pi(t)\,,
	\] where the inequality follows since $\tilde f := \frac{1}{f(\one)}f = \frac{1}{\|f\|} f$ is a state (see Propositions 3.6 and 3.8 \cite{Paulsen}).  
	
	Let $\xi \in B(\mathcal{H})^*$ and write $\xi = \xi_1 - \xi_2 +i\xi_3 - i \xi_4$ be the Jordan decomposition of $\xi$ into positive linear functionals. The following estimate holds: \begin{align*}
    \|(\xi - P(\xi))\circ T_t\| &\le \|(\xi_1 - \xi_1(\one)\pi )\circ T_t\|+\|(-1)(\xi_2 - \xi_2(\one)\pi )\circ T_t\|+\cdots\\
    &\qquad \cdots +\|(i)(\xi_3 - \xi_3(\one)\pi )\circ T_t\|+\|(-i)(\xi_4 - \xi_4(\one)\pi )\circ T_t\|\\
    &\le g_\pi(t)\sum_{j=1}^4 \xi_j(\one) = g_\pi(t) \sum_{j=1}^4 \|\xi_j\| \\
    &\le 2\|\xi\|g_\pi(t) \to 0 \text{ as }t\to \infty.
\end{align*}Thus, $\|\cdot \|_{B(\mathcal{H})^*}-\lim_{t\to \infty} T'_t = P \neq 0$. 

Putting these two facts together with the help of Lemma~\ref{lem:exp_decay} above, we see that there exist constants $c\ge 1$ and $\gamma>0$ for which \begin{equation}
	\|T_t' - P\| \le ce^{-\gamma t},
\end{equation} completing the proof.

$(2. \Rightarrow 1.)$: Obviously, the condition $\|T'_t - P\| \le ce^{-\gamma t}$ implies that $\|\cdot \|-\lim_{t\to \infty} T_t' = P\neq 0$ since $P$ is rank one. Observe that $PT_t' = (\lim T_s')T_t'
	 = \lim T_{t+s} = P = T_t'P$. In particular, $T_t'$ leaves $P(B(\mathcal{H})^*)$ invariant. Notice that $P$ preserves states in the sense that if $\psi$ is a state, so is $P\psi$. This follows because $T_t'\psi$ is a state for all $t$. Hence there must be a state $\pi$ so that $P(B(\mathcal{H})^*)=span(\pi)$, and in fact $P\psi = \pi$ for all other states $\psi$. 
	 
	 Moreover, we note  $T'_t \pi = \pi$ for all $t\ge 0$, hence $\pi$ is a fixed point. To show that $\pi$ is a dynamical dynamical fixed point, note
	\begin{align*}
		|(\psi - \pi)\circ T_t(A)| &\le \|A\| \|T_t'\psi - P\psi \|\\
		&\le \|A\| \|T_t' - P\| \|\psi\|\\
		&\le \|A\|ce^{-\gamma t},
	\end{align*} from which the claim follows easily.
\end{proof}

We turn our attention to understanding how our dynamical fixed point assumption fits into the the Schr\"odinger picture. In particular, recall that a \textit{periodic point} of a Lindbladian $L$, is any $X\in B(\mathcal{H})$ where $LX = i \lambda X$ for some real number $\lambda$. Note that fixed points are therefore periodic points with $\lambda = 0$. Write $P_\lambda$ for the projection onto the subspace associated to the eigenvalue $i\lambda$. In particular, one can define $T_\phi = \sum_{\lambda} P_\lambda$ and $T_{t,\phi} = T_t T_\phi = \sum_{\lambda} e^{i\lambda t}P_\lambda$ (see \cite{Cubitt,Wolf} for details). In particular, if $L$ has a unique fixed point and no other periodic points, then $T_{\phi, t} = P_0$ is just the projection onto the subspace spanned by the fixed point. 

Recall the mixing coefficient $\eta(T^\dagger_t)$  (see \cite{Brandao, Cubitt, Kastoryano_et_al, Temme_et_al}) defined via
	\begin{equation}
		\eta(T^{\dagger}_t) = \sup_{\substack{\varrho \ge 0:\\ \Tr(\varrho) = 1}} \frac{1}{2} \|T^{\dagger}_t \varrho - T^{\dagger}_{\phi, t}\varrho \|_1. 
	\end{equation} 

We note that the mixing coefficient can be defined without reference to a specific model of an open quantum spin system (see \cite{Kastoryano_et_al, Temme_et_al}). We therefore state the following Theorem without reference to a specific model. However, we note that when Theorem~\ref{thm:rapid_mixing} is applied to the case we consider in the main text, $\eta((T_t^{\Lambda})^\dagger)$ in general depends on the volume $\Lambda$. This is in contrast to the assumption in \cite{Brandao, Cubitt} where the authors of both studies require the mixing coefficient to decay at a volume-independent rate.

\begin{thm}\label{thm:rapid_mixing}
	Let $\mathcal{H}$ be a finite dimensional Hilbert space and let $T_t = \exp\{t L\}$ be a quantum dynamical semigroup with generator $L$. Then, the following are logically equivalent. 
	\begin{enumerate}[label = (\alph*)]
		\item The pre-composition semigroup $T_t'$ admits a dynamical fixed point (see Theorem~\ref{thm:exponential_steady_state}) and has no other periodic points. 
		\item The Schr\"odinger picture semigroup $T_t^\dagger$ admits a fixed point with no other periodic points and $\eta(T^\dagger_t) \to 0$ as $t\to \infty$. 
	\end{enumerate}
\end{thm}
\begin{proof}
	Let $S\subset B(\mathcal{H})^*$ denote the set of states. Recall that for every $\psi \in S$, there exists a unique density matrix $\varrho_\psi$ so that $\psi(\cdot) = \Tr[\varrho_\psi (\cdot)]= \Tr[\varrho_\psi^*(\cdot)]$. It is not difficult to show that the Schr\"odinger picture semigroup $T^\dagger_t$ fixes $\varrho_\pi$ if and only if the precomposition semigroup $T'_t$ fixes $\psi$.

Now, consider the following calculation: \begin{align*}
	\sup_{\psi \in S} \sup_{\substack{A\in B(\mathcal{H}):\\ \|A\| \le 1}} |(\psi  - \pi)T_t(A)| & = \sup_{\psi \in S} \sup_{\substack{A\in B(\mathcal{H}):\\ \|A\| \le 1}} |\Tr[ (\varrho_\psi - \varrho_\pi) T_t(A)]| \\
	&= \sup_{\substack{\varrho \ge 0\\ \Tr \varrho =1}} \sup_{\substack{A\in B(\mathcal{H}): \|A\|\le 1}} |\Tr[ (T_t^\dagger \varrho - T^\dagger_t \varrho_\pi) (A)|\\
	&=  \sup_{\substack{\varrho \ge 0\\ \Tr \varrho =1}} \|T^\dagger_t(\varrho) - T^\dagger_t(\varrho_\pi)\|_1\\
	&= \sup_{\substack{\varrho \ge 0:\\ \Tr \varrho = 1}} \|T_t^\dagger (\varrho) - \varrho_\pi\|_1 = 2\eta(T^{\dagger}_t)
\end{align*} where $T^\dagger$ denotes the conjugate transpose. Notice we have used the fact that the trace norm is dual to the operator norm with respect to the trace pairing to obtain the third equality. 

In other words, if there exists a dynamical fixed point for $T_t$ and the Schr\"odinger picture semigroup $T^\dagger _t$ has no other periodic points, then \[ \eta(T_t^\dagger) \le \frac{1}{2} g_\pi(t) \to 0 \text{ as }t\to \infty,\] where $\eta$ is the mixing coefficient defined above. 

To finish the proof, we will show that there is a one-to-one correspondence between the periodic points of $T_t'$ and $T_t^\dagger$. Let $\theta$ be a periodic point of $T_t'$. Then there is real $\lambda \neq 0$ so that $T_t'\theta = e^{i\lambda t}\theta$. By the Riesz representation theorem, there is $X_\theta \in B(\mathcal{H})$ so that $\theta(\cdot) = \Tr[(X_\theta)^* \cdot ]$. In which case for any other $x\in B(\mathcal{H})$, one has 
	\[
		\Tr[(T^\dagger_t(X_\theta))^*x] = \Tr[(X_\theta)^* T_t(x)] = \theta(T_t(x)) = e^{i\lambda t} \theta(x) = \Tr[(e^{-i\lambda t} X_\theta)^* x]. \] Hence $T^\dagger_t(X_\theta) = e^{-i \lambda t} X_\theta$. Running this argument in reverse shows that $\theta \mapsto X_\theta$ is a bijection between the periodic points. 
	\end{proof}
	
	\begin{rmk}
		It is not difficult to show that, as matrices, $T_t'$ is the transpose of $T_t$, where $T_t^\dagger$ is the \textit{conjugate} transpose. Therefore, we could have argued the last paragraph by appealing to elementary linear algebra as well. 
	\end{rmk}

\bibliographystyle{plain}
\bibliography{bib.bib}

\end{document}